\documentclass[11pt]{article}
\usepackage{verbatim,url,enumerate,color,paralist}
\usepackage{amsmath,amsfonts,amsthm}
\usepackage{epsfig,amssymb,amstext,xspace}
\usepackage{algorithm}
\usepackage{algorithmicx,algpseudocode}
\usepackage{fullpage}

\RequirePackage{fix-cm}
\usepackage{graphicx}
\usepackage{subfig}
\usepackage[usenames,dvipsnames]{xcolor}
\usepackage[colorlinks,citecolor=blue,linkcolor=BrickRed]{hyperref}
\usepackage{times}
\usepackage{latexsym,graphicx,epsfig,color,mathdots}
\usepackage{amsfonts,amssymb,amsmath,amstext}
\usepackage{url,setspace}
\usepackage{multirow}
\usepackage{rotating}
\usepackage{tikz}

\usetikzlibrary{decorations.markings}
\usetikzlibrary{arrows}
\tikzstyle{block}=[draw opacity=0.7,line width=1.4cm]
\tikzstyle{graphnode}=[circle, draw, fill=black!20, inner sep=0pt, minimum width=6pt]
\tikzstyle{point}=[circle, draw, fill=black!30, inner sep=0pt, minimum width=1pt]
\tikzstyle{input}=[rectangle, draw, fill=black!75,inner sep=3pt, inner ysep=3pt, minimum width=4pt]
\tikzstyle{unmatched}=[graphnode,fill=black!0]
\tikzstyle{shaded}=[graphnode,fill=black!20]
\tikzstyle{matched}=[graphnode,fill=black!100]  	
\tikzstyle{matching} = [ultra thick]
\tikzset{
    >=stealth',
    pil/.style={
           ->,
           thick,
           shorten <=2pt,
           shorten >=2pt,}
}
\tikzset{->-/.style={decoration={
  markings,
  mark=at position .5 with {\arrow{>}}},postaction={decorate}}}

\newtheorem{theorem}{Theorem}[section]
\newtheorem{lemma}[theorem]{Lemma}
\newtheorem{corollary}[theorem]{Corollary}
\newtheorem{definition}[theorem]{Definition}

\newtheorem{proposition}[theorem]{Proposition}

\newtheorem{fact}[theorem]{Fact}

\newcommand{\IGNORE}[1]{}

\newcommand{\sa}{\textsf{Sherali-Adams}}
\newcommand{\ls}{\textsf{Lov\'asz-Schrijver}}
\newcommand{\la}{\textsf{Lasserre}}

\newcommand{\iLS}{\textsf{LS}}
\newcommand{\iLSp}{\textsf{LS}$_+$}
\newcommand{\iSA}{\textsf{SA}}

\newcommand{\pop}{\mathcal{P}}

\newcommand \reals {\mathbb{R}}
\def \integers {\mathbb{Z}}
\newcommand{\xout}[1]{x\left( \delta^{out}(#1)\right)}
\newcommand{\xin}[1]{x\left( \delta^{in}(#1)\right)}
\newcommand{\atspbal}{\homog{\textsf{ATSP$_{\mathit{BAL}}$}}}
\newcommand{\atspdfj}{\homog{\textsf{ATSP$_{\mathit{DFJ}}$}}}
\newcommand{\atspbalpolytope}{\widehat{\textsf{ATSP}}_{\mathit{BAL}}}
\newcommand{\atspdfjpolytope}{\widehat{\textsf{ATSP}}_{\mathit{DFJ}}}
\newcommand{\PATSP}{\textsf{PATSP}}
\newcommand{\PATSPpolytope}{\widehat{\textsf{PATSP}}}
\newcommand{\tsp}{\textsc{TSP}}
\newcommand{\atsp}{\textsc{ATSP}}
\newcommand{\pathtsp}{\textsc{path~TSP}}
\newcommand{\pathatsp}{\textsc{path~ATSP}}
\renewcommand{\P}{$\mathbf{P}$}
\newcommand{\NP}{$\mathbf{NP}$}

\newcommand\bx {{\bf x}}
\newcommand\by {{y}}
\newcommand\bz {{z}}
\newcommand{\iLP}{\textsf{LP}}

\newcommand{\gdecomp}[1]{\mathcal{C}(#1)}
\newcommand{\gdecomplong}[2]{\mathcal{C}_{#2}(#1)}
\newcommand{\gdecompsymbol}{\widehat{\mathcal{C}}}
\newcommand{\cindset}{\mathcal{N}}
\newcommand{\fracset}{\mathcal{F}}
\newcommand{\notfracset}{\overline{\mathcal{F}}}
\newcommand{\sgn}{\mathcal{I}}
\newcommand{\yvec}[2]{y^{#1,\,#2}}	
\newcommand{\zvec}[2]{y^{#1,\,#2}}	
\newcommand{\zveconly}{y}
\newcommand{\onevec}[2]{\textbf{1}^{#1,\,#2}}		
\newcommand{\goodfrac}[2]{F^{#1}(#2)}	
\newcommand{\szgoodfrac}[2]{f^{#1}(#2)}	
\newcommand{\cindex}[1]{\textit{index}(#1)} 
\newcommand{\indfrac}{h}	
\newcommand{\tour}{\textit{tour\/}}

\newcommand{\opt}{\textit{opt\/}}
\newcommand{\dfj}{\textit{dfj\/}}

\newcommand{\saop}{\mathcal{SA}}
\newcommand{\homog}[1]{{#1}}

\begin{document}

\title{On Integrality Ratios for Asymmetric TSP \\
	in the Sherali-Adams Hierarchy
\footnote{An extended abstract of this work appeared in the proceedings of the 40th International Colloquium on Automata, Languages, and Programming ({ICALP} 2013).}
}


\author{Joseph Cheriyan\thanks{
 Dept. Comb. \& Opt., University of Waterloo, Canada,
 Email:\texttt{$\{$jcheriyan,z9gao,k2georgi$\}$@uwaterloo.ca}} 
\and Zhihan Gao$^\dagger$  
\and Konstantinos Georgiou$^\dagger$ 
\and
Sahil Singla\thanks{School of Comp.\ Sci., Carnegie Mellon University, USA, 
Email:\texttt{ssingla@cmu.edu}
}
}


\maketitle

\begin{abstract}
We study the \atsp\ (Asymmetric Traveling Salesman Problem), and our
focus is on negative results in the framework of the Sherali-Adams
(\iSA) Lift and Project method.


Our main result pertains to the standard LP (linear programming)
relaxation of \atsp, due to Dantzig, Fulkerson, and Johnson.
For any fixed integer $t\ge0$ and small $\epsilon$,
$0<\epsilon\ll{1}$,
there exists a digraph $G$ on
$\nu=\nu(t,\epsilon)=O(t/\epsilon)$
vertices such that the integrality ratio for level~$t$ of the \iSA\ system
starting with the standard LP on $G$ is
$\ge 1+\frac{1-\epsilon}{2t+3} \approx \frac43, \frac65, \frac87, \dots$.
Thus, in terms of the input size,
the result holds for any $t = 0,1,\dots,\Theta(\nu)$ levels.
Our key contribution is to identify
a structural property of digraphs that
allows us to construct fractional feasible solutions
for any level~$t$ of the \iSA\ system starting from the standard~LP.
Our hard instances are simple and satisfy the structural property.

There is a further relaxation of the standard LP called the
balanced LP, and our methods simplify considerably
when the starting LP for the \iSA\ system is the balanced~LP;
in particular, the relevant structural property (of digraphs)
simplifies such that
it is satisfied by the digraphs given by
the well-known construction of Charikar, Goemans and Karloff (CGK).
Consequently, the CGK digraphs serve as hard instances,
and we obtain an integrality ratio of  $1 +\frac{1-\epsilon}{t+1}$
for any level~$t$ of the \iSA\ system,
where $0<\epsilon\ll{1}$
and the number of vertices is
$\nu(t,\epsilon)=O((t/\epsilon)^{(t/\epsilon)})$.

Also, our results for the standard~LP extend to the \pathatsp\ (find a
min cost Hamiltonian dipath from a given source vertex to a given sink
vertex).

\end{abstract}

\section{Introduction}\label{sec:intro}

The Traveling Salesman Problem (\tsp) is a celebrated problem in
combinatorial optimization, with many connections to theory and
practice.  The problem is to find a minimum cost tour of a set of
cities; the tour should visit each city exactly once. The most well
known version of this probelm is the symmetric one (denoted \tsp),
where the distance (a.k.a.~cost) from city~$i$ to city~$j$ is equal to
the distance (cost) from city~$j$ to city~$i$. The more general version
is called the asymmetric TSP (denoted \atsp), and it does not have the symmetry
restriction on the costs. Throughout, we assume that the costs satisfy
the triangle inequalities, i.e., the costs are metric.

Linear programming (LP) relaxations play a central role in solving  \tsp\ 
or \atsp, both in
practice and in the theoretical setting of approximation algorithms.
Many LP relaxations are known for \atsp, see \cite{RT12} for a recent
survey. The most well-known relaxation (and the one that is most useful
for theory and practice) is due to Dantzig, Fulkerson and Johnson; we
call it the standard~LP or the DFJ~LP. It has a constraint for every
nontrivial cut, and has an indegree and an outdegree constraint for
each vertex; see Section~\ref{subsec:atsp-lp}.
There is a further relaxation of the standard~LP that is of interest; we
call it the balanced LP (Bal~LP); it is obtained from the standard~LP
by replacing the indegree and outdegree constraint at each vertex by a
balance (equation) constraint. For metric costs, the optimal value of
the standard~LP is the same as the optimal value of the balanced LP;
this is a well-known fact, see \cite{RT12}, \cite[Footnote~3]{CGK06}.

One key question in the area is the quality of the objective value
computed by the standard~LP.  This is measured by the \textit{integrality
ratio} (a.k.a.~integrality gap) of the relaxation, and is defined to be
the supremum over all instances of the integrality ratio of the
instance.  The integrality ratio of an instance $I$ is given by
$\opt(I)/\dfj(I)$, where $\opt(I)$ denotes the optimum (minimum cost of
a tour) of $I$, and $\dfj(I)$ denotes the optimal value of the standard~LP
relaxation of $I$; we assume that the optima exist and that
$\dfj(I)\not=0$.\footnote{
Although the term integrality ratio is used in two different
senses---one refers to an instance, the other to a relaxation
(i.e., all instances)---the context will resolve the ambiguity.}

For both \tsp\ and \atsp,
significant research efforts have been devoted over several decades
to prove bounds on the integrality ratio of the standard~LP.
For \tsp, methods based on Christofides' algorithm show that
the integrality ratio is $\le\frac32$,
whereas the best lower bound known on the integrality ratio is $\frac43$.
Closing this gap is a major open problem in the area.
For \atsp, a recent result of Asadpour et al.~\cite{AGMOS-soda10}
shows that the integrality ratio is $\le O(\log{n}/\log\log{n})$.
On the other hand, Charikar, et al.~\cite{CGK06} showed a
lower bound of~2 on the integrality ratio,
thereby refuting an earlier conjecture of Carr and Vempala \cite{CV04}
that the integrality ratio is $\leq\frac43$.

Lampis \cite{lampis12} and Papadimitriou and later Vempala
\cite{papad-vempala-06}, have proved
hardness-of-approximation thresholds of $\frac{185}{184}$ for \tsp\ and
$\frac{117}{116}$ for \atsp,  respectively;
both results assume that \P$\not=$\NP.
Recently, Karpinski, et al \cite{KLS-isaac13} have improved both
hardness-of-approximation thresholds to
123/122 and 75/74, respectively, assuming that \P$\not=$\NP.

Our goal is to prove lower bounds on the integrality ratios for \atsp\ 
for the tighter LP relaxations obtained by applying
the Sherali-Adams Lift-and-Project method.
Before stating our results, we present an overview of Lift-and-Project
methods.

\subsection{Hierarchies of convex relaxations}


Over the past 25 years, several methods have been developed in order
to obtain tightenings of relaxations in a systematic manner.
Assume that each variable $y_i$ is in the interval $[0,1]$,
i.e., the integral solutions are zero/one, and
let $n$ denote the number of variables in the original relaxation.
%
%
The goal is to start with a simple relaxation,
and then iteratively obtain a sequence of
stronger/tighter relaxations
such that the associated polytopes form a nested family
that contains (and converges to) the integral hull\footnote{
By the \textit{integral hull} we mean
the convex hull of the zero-one solutions
that are feasible for the original relaxation.}.


These procedures, usually called {\em Lift-and-Project}
hierarchies (or systems, or methods, or procedures),
use polynomial reasonings together with
the fact that in the 0/1 domain, general polynomials can be reduced
to multilinear polynomials (utilizing the identity $y_i^2=y_i$), and
then finally obtain a stronger relaxation by applying linearization
(e.g., for subsets $S$ of $\{1,\dots,n\}$,
the term $\prod_{i \in S} y_i$ is replaced by a variable $y_S$).
In this overview, we gloss over the Project step.
In particular,
Sherali and Adams~\cite{SA90} devised the \sa\ (\iSA) system,
Lov\'asz and Schrijver~\cite{LS91}
devised the \ls\ (\iLS) system,
and Lasserre~\cite{las02} devised the \la\ system.
See Laurent~\cite{Lau03} for a survey of these systems;
several other Lift-and-Project systems are known,
see~\cite{CM-chapter,AT-ipco11}.

The index of each relaxation in the sequence of tightened relaxations
is known as the \textit{level} in the hierarchy; the level of the
original relaxation is defined to be zero.
For each of these hierarchies and for any $t=O(1)$,
it is known that
the relaxation at level $t$ of
the hierarchy can be solved to optimality in polynomial time,
assuming that the original relaxation has
a polynomial-time separation oracle,
\cite{tourlakisthesis06}
%
%
(additional mild conditions may be needed for some hierarchies).
%
%
%
In fact, the relaxation at level $n$ is exact, i.e.,
the associated polytope is equal to the integral hull.
%

Over the last two decades,
a number of important improvements on approximation guarantees have
been achieved based on relaxations obtained from Lift-and-Project
systems.
See \cite{CM-chapter} for a recent survey of many such positive results.
%

Starting with the work of Arora et al.~\cite{ABLT06},
substantial research efforts have been devoted to
showing that tightened relaxations (for many levels) fail to reduce
the integrality ratio for many combinatorial optimization problems
(see~\cite{CM-chapter} for a list of negative results).
This task seems especially difficult for
the \iSA\ system because it strengthens
relaxations in a ``global manner;''
this enhances its algorithmic leverage for deriving positive results,
but makes it more challenging
to design instances with bad integrality~ratios.
%
%
Moreover, an integrality ratio for the \iSA\ system may be viewed as an
unconditional inapproximability result for a restricted model of
computation, whereas, hardness-of-approximation results are usually
proved under some complexity assumptions, such as {\P}$\not=${\NP}.
%
%
The \iSA\ system is known to be more powerful than
the \iLS\ system, while it is weaker than the \la\ system;
it is incomparable with the \iLSp\ system
(the positive-semidefinite version of the \ls\ system \cite{LS91}).


A key paper by Fern\'andez de la Vega and Kenyon-Mathieu \cite{FK07}
introduced a probabilistic interpretation of the \iSA\ system, and
based on this, negative results (for the \iSA\ system) have been proved
for a number of combinatorial problems;
also see
Charikar et al.~\cite{CMM09}, and Benabbas, et al.~\cite{BCGM11}.
At the moment, it is not clear that methods based on \cite{FK07} could
give negative results for \tsp\ and its variants, because the natural
LP relaxations (of \tsp\ and related problems) have ``global
constraints.''



To the best of our knowledge, there are only two previous papers with
negative results for Lift-and-Project methods applied to \tsp\ and its
variants.
Cheung~\cite{cheung05} proves an integrality ratio of
$\frac43$ for \tsp, for $O(1)$ levels of \iLSp.
For \atsp, Watson~\cite{watson11} proves an integrality ratio of
$\frac32$ for level~1 of the \ls\ hieararchy, starting from the
balanced~LP (in fact, both the hierarchies \iLS\ and \iSA\ give the
same relaxation at level~one).

We mention that Cheung's results~\cite{cheung05}
for \tsp\ do not apply to \atsp, although
at level~0, it is well known that
any integrality ratio for the standard~LP for \tsp\
applies also to the standard~LP for \atsp\ 
(this relationship does not hold for level~1 or higher).

\subsection{Our results and their significance}


Our main contribution
is a generic construction of
fractional feasible solutions
for any level~$t$ of the \iSA\ system starting from
the standard~LP relaxation of \atsp.
We have a similar but considerably simpler construction
when the starting LP for the \iSA\ system
is the balanced~LP.
Our results on integrality ratios are direct corollaries.

We have the following results pertaining to
the balanced~LP relaxation of \atsp:
We formulate a property of digraphs that we call
the good decomposition property, and
given any digraph with this property,
we construct a vector $y$ on the edges such that
$y$ is a fractional feasible solution to
the level~$t$ tightening of the balanced~LP by the \sa\ system.
Charikar, Goemans, and Karloff (CGK) \cite{CGK06}
constructed a family of digraphs for which
the balanced~LP has an integrality ratio of~2.
We show that the digraphs in the CGK family
have the good decomposition property,
hence, we obtain an integrality ratio for level~$t$ of \iSA.
In more detail,
we prove that for any integer $t\ge0$ and small enough $\epsilon>0$,
there is a digraph $G$ from the CGK family
on $\nu=\nu(t,\epsilon)=O((t/\epsilon)^{t/\epsilon})$
vertices such that the integrality ratio of the level-$t$ tightening of
Bal LP is at least
$1+\frac{1-\epsilon}{t+1} \approx 2, \frac32, \frac43, \frac54,
\dots$ (where $t=0$ identifies the original relaxation).

Our main result pertains to the standard~LP relaxation of \atsp.
Our key contribution is to identify
a structural property of digraphs that
allows us to construct fractional feasible solutions
for the level~$t$ tightening of the standard~LP by the \sa\ system.
This construction is much more difficult than the
construction for the balanced~LP.
We present a simple family of digraphs that
satisfy the structural property, and this immediately
gives our results on integrality ratios.
We prove that for any integer $t\ge0$ and small enough $\epsilon>0$,
there are digraphs $G$ on $\nu=\nu(t,\epsilon)=O(t/\epsilon)$ vertices
such that the integrality ratio of
the level~$t$ tightening of the standard~LP on $G$
is at least $1+\frac{1-\epsilon}{2t+3} \approx \frac43, \frac65,
\frac87, \frac{10}9, \dots$.
The \textit{rank} of a starting relaxation (or polytope)
is defined to be the minimum number of tightenings
required to find the integral hull (in the worst case).
An immediate corollary is that
the \iSA-rank of the standard~LP relaxation on a digraph $G=(V,E)$
is at least linear in $|V|$, whereas,
the rank in terms of the number of edges is $\Omega(\sqrt{|E|})$
(since the LP is on a complete digraph, namely, the metric completion).

Our results for the balanced~LP and for the standard~LP
are incomparable, because
the \iSA\ system starting from the standard~LP
is strictly stronger than
the \iSA\ system starting from the balanced~LP,
although both the level~zero LPs have the
same optimal value, assuming metric costs.
(In fact, there is an example on 5~vertices
\cite[Figure~4.4,~p.60]{paulthesis08} such that
the optimal values of the level~1 tightenings are different:
$9\frac13$ for the balanced~LP and $10$ for the standard~LP.)

Finally, we extend our main results to the natural relaxation of
\pathatsp\ (min cost Hamiltonian dipath from a given source vertex to a
given sink vertex), and
we obtain integrality ratios $\ge 1+\frac{2-\epsilon}{3t+4} \approx
\frac32, \frac97, \frac65, \frac{15}{13}, \dots$ for the level-$t$
\iSA\ tightenings.
Our result on \pathatsp\ is obtained by ``reducing'' from the result
for \atsp; the idea behind this comes from an analogous result of
Watson \cite{watson11} in the symmetric setting; Watson gives a method
for transforming Cheung's \cite{cheung05} result on the integrality
ratio for \tsp\ to obtain a lower bound on the integrality ratio for
\pathtsp.

The solutions given by our constructions are \textit{not}
positive semidefinite;
thus, they do not apply to the \iLSp\ hierarchy
nor to the \la\ hierarchy.

Let us assess our results, and place them in context.
Observe that our integrality ratios fade out as the
level of the \iSA\ tightening increases,
and for $t\ge 35$ (roughly)
our integrality ratio falls below the hardness~threshold of
$\frac{75}{74}$ of \cite{KLS-isaac13}.
Thus, our integrality ratios cannot be optimal, and it is possible
that an integrality ratio of~2 can be proved for $O(1)$ levels
of the \iSA\ system.

On the other hand,  our results are not restricted to $t=O(1)$.
For example, parameterized with respect to the number of vertices in
the input $\nu$, our lower bound for the standard~LP holds even for
level~$t=\Omega(\nu)$, and our lower bound for the balanced~LP
(which improves on our lower bound for the standard~LP)
holds even for level~$t=\Omega(\log\nu / \log\log\nu)$, thus giving
unconditional inapproximability results for these restricted
algorithms, even allowing super-polynomial running time.

Moreover, our results (and the fact that they are not optimal)
should be contrasted with the
known integrality~ratio results for the level~zero standard~LP,
a topic that has been studied for decades.



\section{Preliminaries}

When discussing a digraph (directed graph), we
use the terms dicycle (directed cycle), etc.,
but we use the term edge rather than directed edge or arc.
For a digraph $G=(V,E)$ and $U\subseteq{V}$,
$\delta^{out}(U)$ denotes $\{ (v,w)\in{E}: v\in U, w\not\in U\}$,
the set of edges outgoing from $U$,
and $\delta^{in}(U)$ denotes $\{ (v,w)\in{E}: v\not\in U, w\in U\}$.
For $x\in\reals^{E}$ and $S\subseteq{E}$,
$x(S)$ denotes $\sum_{e\in{S}}x_e$.

By the \textit{metric completion} of a
digraph $G=(V,E)$ with nonnegative edge costs $c\in\reals^E$,
we mean the complete digraph $G'$ on $V$
with the edge costs $c'$,
where $c'(v,w)$ is taken to be the minimum cost (w.r.t. $c$)
of a $v,w$ dipath of $G$.

An \textit{Eulerian subdigraph} of $G$ is defined as follows:
the vertex set is $V$
and the edge set is a ``multi-subset'' of $E$
(that is, each edge in $E$ occurs zero or more times)
such that (i) the indegree of every vertex equals its outdegree, and
(ii) the subdigraph is weakly connected (i.e., the underlying
undirected graph is connected). The ATSP on the metric completion $G'$
of $G$ is equivalent to finding a minimum cost Eulerian subdigraph of
$G$.

For a positive integer $t$ and a ground set $U$, let $\pop_{t}$ denote the
family of subsets of $U$ of size at most $t$, i.e.,
$\pop_{t} = \{ S :  S \subseteq U, |S| \leq t\}$.
We usually take the ground set to be the set of edges of a fixed digraph.
Now, let $G$ be a digraph, and
let the ground set (for $\pop_{t}$) be $E=E(G)$.
Let $E'$ be a subset of $E$.
Let $\onevec{E'}{t}$ denote a vector indexed by elements of $\pop_{t}$
such that for any $S \in \pop_t$, $\onevec{E'}{t}_S=1$ if $S\subseteq{E'}$,
and $\onevec{E'}{t}_S = 0$, otherwise.
Note that $\onevec{E'}{1}$ has the entry for $\emptyset$ at 1,
and the other entries give the incidence vector of $E'$.
%

We denote set difference by $-$, and
we denote the addition (removal) of a single item $e$
to (from) a set $S$ by $S+e$ (respectively, $S-e$),
rather than by $S\cup\{e\}$ (respectively, $S-\{e\}$).


\subsection{LP relaxations for Asymmetric \tsp}\label{sec:atsp-relaxations}
\label{subsec:atsp-lp}

Let $G=(V,E)$ be a digraph with nonnegative edge costs $c$. Let
$\atspdfjpolytope(G)$ be the feasible region (polytope) of the following
linear program that has a variable $x_e$ for each edge $e$ of $G$:
\begin{align}
\textup{minimize} & \sum_{e} c_e x_e & \notag \\
\textup{subject to} \quad
  & \xin{S} \geq 1,& \quad \forall S:~\emptyset \subset S \subset V \notag  \\
  & \xout{S} \geq 1,& \quad \forall S:~\emptyset \subset S \subset V \notag  \\
  & \xin{\{v\}}  = 1, \quad \xout{\{v\}} = 1, & \quad \forall v \in V \notag  \\
  & \textbf{0} \leq \bx \leq \textbf{1} & \notag
\end{align}
In particular, when $G$ is a complete digraph with metric costs, the
above linear program is the standard~LP relaxation of ATSP
(a.k.a.~DFJ~LP).

We obtain the balanced~LP (Bal~LP) from the standard~LP by
replacing the two constraints
$\xin{\{v\}}  = 1,\; \xout{\{v\}} = 1$
by the constraint
$\xin{\{v\}}  = \xout{\{v\}}$,
for each vertex $v$.
Let $\atspbalpolytope(G)$ be the feasible region (polytope) of Bal~LP.
\begin{align}
\textup{minimize} & \sum_{e} c_e x_e & \notag \\
\textup{subject to} && \notag \\
  & \xin{S} \geq 1,& \quad \forall S:~\emptyset \subset S \subset V \notag  \\
  & \xout{S} \geq 1,& \quad \forall S:~\emptyset \subset S \subset V \notag  \\
  & \xout{\{v\}} =  \xin{\{v\}},& \quad \forall v \in V \notag  \\
  & \textbf{0} \leq \bx \leq \textbf{1} & \notag
\end{align}
In particular, when $G$ is a complete digraph with metric costs, the
above linear program is the balanced~LP relaxation of ATSP.

Our construction of fractional feasible solutions exploits the
structure of the original digraph.
This is the reason for discussing the polytopes on the original digraph
(and not only on the complete digraph).
To justify this, we observe that any feasible solution for the original
digraph can be extended to a feasible solution for the complete digraph
by ``padding with zeros.'' (This argument is formalized in Section
\ref{sec:EVto0}).



\subsection{The \sa\ system}

\begin{definition}[The \sa\ system]\label{def:SA-definition}
Consider a polytope $\widehat{P} \subseteq [0,1]^n$ over the variables
$y_1,\ldots, y_n$, and its description by a system of linear
constraints of the form $\sum_{i=1}^n a_i y_i \geq{b}$;
note that the constraints $y_i\ge0$ and $y_i\le{1}$
for all $i\in\{1,\dots,n\}$
are included in the system.
The level-$t$ \sa\ tightened relaxation $\saop^t(\widehat{P})$
of $\widehat{P}$, is an \iLP\ over the variables
 $\{y_S\;:\; S\subseteq \{1,2,\ldots,n\},\; |S|\le{t+1} \}$
(thus, $y\in\reals^{\pop_{t+1}}$ where
$\pop_{t+1}$ has ground~set $\{1,2,\ldots,n\}$);
%
%
moreover, we have $y_{\emptyset} =1$.
For every constraint $\sum_{i=1}^n a_i y_i\geq{b}$
of $\widehat{P}$ and for every disjoint $S,Q \subseteq \{1,\ldots, n\}$
with $|S|+|Q|\leq t$, the following is a constraint of the
level-$t$ \sa\ relaxation.
\begin{equation}\label{equa:sa-constraint}
\sum_{i=1}^n a_i \sum_{\emptyset \subseteq T \subseteq Q}
	(-1)^{|T|} y_{S \cup T \cup \{i\} }
  \geq b \sum_{\emptyset \subseteq T \subseteq Q}
	(-1)^{|T|} y_{S \cup T }.
\end{equation}
\end{definition}

We will use a convenient abbreviation:
  \[
  z_{S,Q} := \sum_{\emptyset \subseteq T \subseteq Q} (-1)^{|T|} y_{S \cup T },
  \]
where $z_{S,Q}$ are auxiliary variables between 0 and 1. 

Informally speaking, the level-$t$ \sa\ relaxation is derived by
multiplying any constraint of the original relaxation by the high
degree polynomial 
$$\prod_{j \in S} y_i \prod_{j \in Q} (1-y_i),$$
where
$S,Q$ are disjoint subsets of $\{1,\ldots, n\}$ with $|S|+|Q|\leq t$.
After expanding the products, we obtain a polynomial of degree at most
$t+1$. Replacing any occurrences of $\prod_{i \in S} y_i$ by the
corresponding variable $y_S$ for all $S\subseteq\{1,\ldots, n\}$ gives
the constraint described in Inequality~\eqref{equa:sa-constraint}
(Definition~\ref{def:SA-definition}).

There are a number of approaches for certifying that $\by \in
\saop^t(\widehat{P})$ for a given $\by$. One popular approach is to give a
probabilistic interpretation to the entries of $\by$, satisfying
certain conditions. We follow an alternative approach, that is
standard, see \cite{Lau03}, \cite[Lemma~2.9]{tourlakisthesis06},
but has been rarely used in the context of integrality ratios.

First, we introduce some notation.
Given a polytope $\widehat{P} \subseteq [0,1]^n$,
consider the cone
$\homog{P}= \{ y_\emptyset (1, \by):~ y_\emptyset\ge0, \by \in \widehat{P}\}$.
(Throughout the paper, we use an accented symbol to denote a polytope,
e.g., $\widehat{P}$, and
the symbol (without accent) to denote the associated cone, e.g., $P$.)
It is not difficult to see
that the \iSA\ system can be applied to the cone ${P}$, so that
the projection in the $n$ original variables
can be obtained by projecting
any $y\in\saop^t({P})$ with $y_\emptyset = 1$
on the $n$ original variables.
Note that $\saop^t({P})$ is a cone, hence,
we may have $y\in\saop^t({P})$ with $y_\emptyset \not= 1$;
but if $y_\emptyset \not=0$,
we can replace $y$ by $\frac{1}{y_\emptyset} y$.
Also, note that $\saop^t(\widehat{P})=\{y:~ y_\emptyset=1, y\in
\saop^t({P})\}$ by Definition~\ref{def:SA-definition}.

For a vector $\by$ indexed by subsets of $\{1,\ldots,n\}$ of size at
most $t+1$, define a \textit{shift operator} ``$*$'' as follows:
for every $e \in \{1,\ldots,n\}$, let $e*\by$ to be a vector indexed by
subsets of $\{1,\ldots,n\}$ of size at most $t$, such that $(e*y)_S
:=y_{S + e}$.
We have the following folklore fact, \cite[Lemma~2.9]{tourlakisthesis06}.

\begin{fact}.~ \label{fact:recursive-SA}
$\by \in \saop^t({P}) ~~\textbf{if and only if}~~
e*\by \in \saop^{t-1}({P}), ~~\textbf{and}~~
\by - e*\by \in \saop^{t-1}({P}),
~~\forall e \in \{1,\ldots,n\}$.
\end{fact}

The reader familiar with the \ls\ system may recognize the similarity
of its definition with the characterization of the \sa\ system of
Fact~\ref{fact:recursive-SA}. In fact, the \iSA\ system differs from
the \iLS\ system only in that it imposes additional consistency
conditions; namely, the moment vector $\by$, indexed by subsets of size
$t+1$, has to be fixed beforehand.
%
%
This seemingly small detail gives the \iSA\ system enhanced power
compared to the \iLS\ system.
%

\subsubsection{Eliminating Variables to 0}\label{sec:EVto0}

In our discussion of the standard~LP and the balanced~LP,
it will be convenient to restrict
the support to the edge set of a given digraph
rather than the complete digraph.
%
%
Thus, we assume that some of the variables are absent.
Formally, this is equivalent to setting
these variables in advance to zero.
As long as the nonzero variables induce a feasible solution, we are
justified in setting the other variables to zero.
%
%
The following result formalizes the arguments.

\begin{proposition}
Let $\widehat{P}$ be the feasible region (polytope) of a linear program.
Let $C$ be a set of indices (of the variables) that does not contain
the support of any ``positive constraint'' of $\widehat{P}$, where a
constraint $\sum_{i=1}^n a_i y_i\geq{b}$ of $\widehat{P}$ is called
positive if $b>0$.
Let $\widehat{P}_C$ be the feasible region (polytope) of the linear
program obtained by removing all variables with indices in $C$ from the
constraints of the linear program of $\widehat{P}$ (informally, the
new LP fixes all variables with indices in $C$ at zero).
Then, for the \iSA\ system,
for any feasible solution $\by$ to the level-$t$ tightening of
$\widehat{P}_C$, there exists a feasible solution $\by'$ to the
level-$t$ tightening of $\widehat{P}$; moreover, $\by'$ is obtained from
$\by$ by fixing variables,
 indexed by subsets 
intersecting $C$, to zero.
%
%
\end{proposition}
\begin{proof}
For $\by \in  \saop^t(\widehat{P}_C)$, the ``extension'' $\by'$ of
$\by$ is defined as follows:
$$
\by'_S = \left\{
\begin{array}{ll}
\by_S &, \textrm{~if~} S\cap C  = \emptyset \\
0  &, \textrm{~otherwise~} \\
\end{array}
\right. 
$$
For the corresponding auxiliary variables $\bz$, this would imply that
$$
\bz_{S,Q}' = \left\{
\begin{array}{ll}
0  &, \textrm{~if~} S\cap C  \not = \emptyset \\
\bz_{S,Q-C}  &, \textrm{~otherwise~}. \\
\end{array}
\right. 
$$
In order to show that $\by' \in  \saop^t(\widehat{P})$, we need to
verify that for every pair of sets $S,Q$ as in
Definition~\ref{def:SA-definition}, we have
$\sum_{i=1}^n a_i   z_{S \cup \{i\},Q}' \geq b z_{S,Q}'$.

First we note that if $S \cap C \not = \emptyset$, then for every
$i$ we have $\bz'_{S \cup \{i\},Q} = \bz'_{S ,Q} = 0$, and hence the
constraint is satisfied trivially.

For the remaining case $S \cap C = \emptyset$, we have
\begin{eqnarray}
\sum_{i=1}^n a_i   z_{S \cup \{i\},Q}'
&=&
\sum_{i \in C} a_i   z_{S \cup \{i\},Q}' + \sum_{i \not \in C} a_i   z_{S \cup \{i\},Q}' \notag\\
&=&
\sum_{i \not \in C} a_i   z_{S \cup \{i\},Q}' \notag \\
&=&
\sum_{i \not \in C} a_i   z_{S \cup \{i\},Q-C} \notag \\
&\geq &
b  ~ z_{S,Q-C} \label{equa:feasibility-smaller-polytope} \\
&=&
b  ~ z_{S,Q-C}' \notag \\
&=&
b ~  z_{S,Q}' \notag,
\end{eqnarray}
where~\eqref{equa:feasibility-smaller-polytope} follows from the
validity of the corresponding constraint of $\widehat{P}_C$;
here, we use the fact that $C$ does not contain the support of any
positive constraint -- otherwise, the summation $\sum_{i\not\in{C}}(\dots)$
would be zero since the index set $\{i: i\not\in{C}\}$ would be empty,
and hence, the inequality
$0 = \sum_{i\not\in{C}}(\dots) \geq b\; z_{S,Q-C}$
would fail to hold for $b>0$ and $z_{S,Q-C} > 0$.
\end{proof}



\section{\iSA\ applied to the Balanced~LP relaxation of ATSP}
\subsection{Certifying a feasible solution} \label{section:balanced}

\begin{figure}[htb]
{
\centerline{
\begin{tikzpicture}[thin,scale=1.0]
\foreach \x in {0,1,2,...,8}
{
	\node at (\x , 0) [graphnode]{};
	\node at (\x , 1) [graphnode]{};
	\node at (\x , 2) [graphnode]{};
}
\foreach \x in {0,1,2,...,7}
{
	\draw [->, very thick] (\x +0.1 , 0) to node[auto]{$ $} (\x + 1-0.1 , 0);
	\draw [<-, very thick] (\x +0.1, 2) to node[auto]{$ $} (\x + 1-0.1 , 2);
}
\foreach \y in {0,1}
{
	\draw [<-, very thick] (0 , \y +0.1) to node[auto]{$ $} (0 , \y +1-0.1);
	\draw [->, very thick] (8 , \y +0.1) to node[auto]{$ $} (8 , \y +1-0.1);
}
\foreach \x in {0,1,2,...,7}
{
	\draw [->-] (\x +0.1 , 1+0.1) .. controls(\x + 0.4, 1.25) .. node[auto]{$ $} (\x + 1-0.1 , 1+0.1);
	\draw [->-] (\x + 1-0.1 , 1-0.1) .. controls(\x + 0.4, 0.75) .. node[auto]{$ $} (\x +0.1 , 1-0.1);
}
\end{tikzpicture}
}
\caption{ A digraph $G$ with a good decomposition given by the dicycle
with thick edges, and the length~2 dicycles $C_j$ formed by the
anti-parallel pairs of thin edges; $G-E(C_j)$ is strongly connected
for each dicycle $C_j$.}
%
\label{fig-bal}
}
\end{figure}
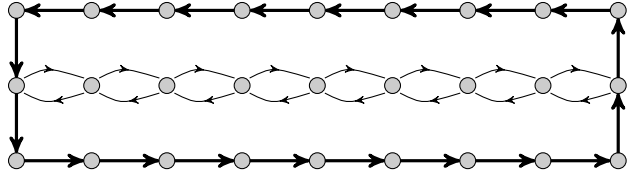

A strongly connected digraph $G=(V,E)$ is said to have
a \textit{good~decomposition} with
\textit{witness~set} $\fracset$
if the following hold
\begin{itemize}
\item[(i)]
$E$ partitions into edge-disjoint dicycles $C_1,C_2,\ldots,C_N$,
that is, there exist edge-disjoint dicycles $C_1,C_2,\ldots,C_N$
such that $E = \bigcup_{1\leq j \leq N} E(C_j)$;
let $\cindset$ denote the set of indices of these dicycles, thus
$\cindset=\{1,\dots,N\}$;
\item[(ii)]
moreover, there exists a nonempty subset $\fracset$ of $\cindset$
such that for each $j\in\fracset$
the digraph $G-E(C_j)$ is strongly connected.
\end{itemize}

Let $\notfracset$ denote $\cindset-\fracset$.
For an edge $e$, we use $\cindex{e}$ to denote the index $j$ of the
dicycle $C_j, j\in\cindset$ that contains $e$.
In this section, by a dicycle $C_i, C_j,$ etc.,
we mean one of the dicycles $C_1,\dots,C_N$, and
we identify a dicycle $C_j$ with its edge set, $E(C_j)$.
See Figure~\ref{fig-bal} for an illustration
of a good decomposition of a digraph.

Informally speaking, our plan is as follows:
for digraph $G$ that has a good decomposition with witness set $\fracset$,
we construct a feasible solution to $\saop^t(\homog{\atspbalpolytope(G)})$
by assigning the same fractional value to the edges of
the dicycles $C_j$ with $j\in\fracset$,
while assigning the value $1$
to the edges of the dicycles $C_i$ with $i\in\notfracset$
(this is not completely correct; we will refine this plan).
Let $\atspbal(G)$ be the associated cone of $\atspbalpolytope(G)$.


\begin{definition}
\label{def:goodfrac}
Let $t$ be a nonnegative integer.
For any set $S\subseteq E$ of size $\leq t+1$, and
any subset $\sgn$ of $\fracset$,
let $\goodfrac{\sgn}{S}$ denote
the set of indices $j\in \fracset-\sgn$ such that
$E(C_j) \cap S \not=\emptyset$;
moreover, let $\szgoodfrac{\sgn}{S}$ denote $|\goodfrac{\sgn}{S}|$,
namely, the number of dicycles $C_j$ with indices in $\fracset-\sgn$
that intersect $S$.
\end{definition}


\begin{definition}
\label{def:balanced}
For a nonnegative integer $t$ and
for any subset $\sgn$ of $\fracset$,
let $\yvec{\sgn}{t}$ be a vector indexed by the elements of
$\mathcal{P}_{t+1}$ and defined as follows:
\[
\yvec{\sgn}{t}_S  = \frac{t+2-\szgoodfrac{\sgn}{S}}{t+2},
	\quad \forall S \in \mathcal{P}_{t+1}
\]
%
\end{definition}

\begin{theorem}
\label{thm:balanced}
Let $G=(V,E)$ be a strongly connected digraph that
has a good decomposition,
and let $\fracset$ be the witness~set. Then
\[ \yvec{\sgn}{t} \in \saop^t(\homog{\atspbalpolytope(G)}), \quad
	\forall t\in\integers_+, \forall \sgn \subseteq \fracset. \]
\end{theorem}
In order to prove our integrality ratio result for $\atspbalpolytope$, we will invoke Theorem~\ref{thm:balanced} for $\sgn=\emptyset$ (the more general setting of the theorem is essential for our induction proof; we give a high-level explanation in the last paragraph of the proof of Theorem~\ref{thm:balanced} below). 
Since also only the values of $\yvec{\emptyset}{t}$ indexed at singleton edges affect the integrality ratio, 
it is worthwhile to summarize all relevant quantities in the next corollary. 


\begin{corollary}
\label{coro:balanced-empty}
We have
\[ \yvec{\emptyset}{t} \in \saop^t(\homog{\atspbalpolytope(G)}), \quad \forall t\in\integers_+.
\]
Moreover, 
for each dicycle $C_j$, $j\in\cindset$, and each edge $e$ of $C_j$ we have
\begin{align}
\yvec{\emptyset}{t}_e =
 \begin{cases}
  \frac{t+1}{t+2},	&\text{if } j \in \fracset\\
  {1},			&\text{otherwise}.
 \end{cases}
\end{align}
\end{corollary}


Informally speaking, we assign the value~1
(rather than a fractional value) to the edges of
the dicycles $C_j$ with $j\in\sgn\subseteq\fracset$.
For the sake of exposition,
we call the dicycles $C_j$ with $j\in\fracset-\sgn$
the \emph{fractional dicycles}, and
we call the remaining dicycles $C_i$
(thus $i\in\sgn\cup\notfracset$) the \emph{integral dicycles}.

\begin{proof}[Proof of Theorem \ref{thm:balanced}:~]
To prove Theorem \ref{thm:balanced}, we need to prove
$$\yvec{\sgn}{t} \in \saop^{t}(\homog{\atspbal(G)}).$$
We prove this by induction on $t$.

Note that $\yvec{\sgn}{t}_\emptyset =1$ by Definition \ref{def:balanced}.

The induction basis is important, and it follows easily from the
good decomposition property.
In Lemma~\ref{lemma:base} (below) we show that
$\yvec{\emptyset}{0} \in \saop^0(\homog{\atspbal(G)})$. 
We conclude that $\yvec{\sgn}{0}$ satisfies the first two sets of constraints of 
$\atspbal(G)$, since $\yvec{\sgn}{0} \geq \yvec{\emptyset}{0}$
(this follows from Definitions~\ref{def:goodfrac},\ref{def:balanced}, since
$\goodfrac{\sgn}{S} \subseteq \goodfrac{\emptyset}{S}$). As for the balance constraints, it is enough to observe that every vertex of our instance (see Figure~\ref{fig-bal}) is incident to pairs of outgoing and ingoing edges, which due to Definition~\ref{def:balanced} are assigned the same value. Finally, again by Definition~\ref{def:balanced}, and for all edges $e$, we have $0\leq \yvec{\sgn}{0}_e \leq 1$. All the above imply that $\yvec{\sgn}{0} \in \saop^0(\homog{\atspbal(G)}),\;\forall \sgn\subseteq\fracset$, as wanted. 

In the induction step, we assume that
$\yvec{\sgn}{t} \in \saop^t(\homog{\atspbal(G)})$ for some integer $t\ge0$
(the induction hypothesis),
and we apply the recursive definition based
on the shift operator,
namely,
$\yvec{\sgn}{t+1} \in \saop^{t+1}(\homog{\atspbal(G)})$ iff for each $e\in E$
\begin{align}
	e * \yvec{\sgn}{t+1} \in \saop^t(\homog{\atspbal(G)}),
	\label{eq:cond1} \\
	\yvec{\sgn}{t+1} - e*\yvec{\sgn}{t+1} \in
	\saop^t(\homog{\atspbal(G)}).  \label{eq:cond2}
\end{align}
Lemma~\ref{lemma:cond1} (below) proves (\ref{eq:cond1}) and
Lemma~\ref{lemma:cond2} (below) proves (\ref{eq:cond2}).
\end{proof}

We prove that $e*\yvec{\sgn}{t+1}$ is in $\saop^t(\homog{\atspbal(G)})$
by showing that for some edges $e$,
$e*\yvec{\sgn}{t+1}$ is a scalar multiple of  $\yvec{\sgn'}{t}$,
where $\sgn'\supsetneq\sgn$
(see Equation~\eqref{eqn:bal-case1} in Lemma~\ref{lemma:cond1});
thus, the induction hinges on the use of $\sgn$.


Before proving Lemma~\ref{lemma:cond1} and Lemma~\ref{lemma:cond2},
we show that $\yvec{\sgn}{t+1}$, restricted to $\pop_{t+1}$,
can be written as a convex combination
of $\yvec{\sgn}{t}$ and the integral feasible solution $\onevec{E}{t+1}$.
This is used in the proof of Lemma~\ref{lemma:cond1};
for some of the edges $e\in E$,
we show that $e * \yvec{\sgn}{t+1} = \yvec{\sgn}{t+1}$
(see Equation~\eqref{eqn:bal-case1}),
and then we have to show that the latter is in $\saop^t(\homog{\atspbal(G)})$.

\begin{fact}.~ \label{fact:recursive}
Let $t$ be a nonnegative integer and let $\sgn$ be a subset of $\fracset$.
Then for any $S\in\pop_{t+1}$ we have \qquad
$\displaystyle \yvec{\sgn}{t+1}_S ~=~ \frac{t+2}{t+3} \yvec{\sgn}{t}_S +
	\frac{1}{t+3}\onevec{E}{t+1}_S.
$
\end{fact}
\begin{proof}
We have $S\subseteq E$, $|S|\leq t+1$, and
we get $\onevec{E}{t+1}_S=1$ from the definition. Thus,
\[ \yvec{\sgn}{t+1}_S ~=~ \frac{t+3-\szgoodfrac{\sgn}{S}}{t+3}
	~=~ \frac{t+2-\szgoodfrac{\sgn}{S}}{t+3}+\frac{1}{t+3}
	~=~ \frac{t+2}{t+3} \yvec{\sgn}{t}_S+\frac{1}{t+3}\onevec{E}{t+1}_S.
\]
\end{proof}

\begin{lemma}
\label{lemma:cond1}
Suppose that $\yvec{\sgn}{t} \in \saop^t(\homog{\atspbal(G)})$,
for each $\sgn\subseteq\fracset$.
Then for all $e\in E$ and for all $\sgn\subseteq\fracset$ we have
$ \displaystyle
	e*\yvec{\sgn}{t+1}  \in \saop^t(\homog{\atspbal(G)})$
\end{lemma}
\begin{proof}
For any $S \in \pop_{t+1}$, the definition of the shift operator gives
$\displaystyle 
	(e*\yvec{\sgn}{t+1})_S = \yvec{\sgn}{t+1}_{S+e}.
$
Let $C(e)$ denote the dicycle containing edge $e$, and recall
that $\cindex{e}$ denotes the index of $C(e)$.

We first show that
\begin{align}
\label{eqn:bal-case1}
 e*\yvec{\sgn}{t+1}_S =
 \begin{cases}
  \frac{t+2}{t+3} \yvec{\sgn+\cindex{e}}{t}_S
	&\text{if } \cindex{e} \in \fracset - \sgn \\
  \yvec{\sgn}{t+1}_S	&\text{otherwise}
 \end{cases}
\end{align}
If $\cindex{e} \in \sgn \cup \notfracset$, that is, the dicycle $C(e)$
is not ``fractional,'' then Definition~\ref{def:balanced} directly
gives $\yvec{\sgn}{t+1}_{S+e} = \yvec{\sgn}{t+1}_S$.
Otherwise, if $\cindex{e} \in \fracset-\sgn$, then
from Definition~\ref{def:balanced} we see that
if $C(e)\cap S\neq \emptyset$, then
$\goodfrac{\sgn}{S+{e}} = \goodfrac{\sgn}{S}$,
and otherwise,
$\szgoodfrac{\sgn}{S+{e}} = \szgoodfrac{\sgn}{S} + 1$.
Hence,

\begin{align}
 (e*\yvec{\sgn}{t+1})_S & =
 \begin{cases}
  \frac{t+3 - \szgoodfrac{\sgn}{S}}{t+3}  &\text{if } C(e)\cap S\neq \emptyset \\
  \frac{t+2 - \szgoodfrac{\sgn}{S}}{t+3}  &\text{if } C(e)\cap S=\emptyset
 \end{cases} \label{eqn:shift} \\
 &= \frac{t+2}{t+3} \yvec{\sgn+\cindex{e}}{t}_S
\end{align}
%
%
where in the last line we use Definition~\ref{def:balanced} to infer
that
 $\szgoodfrac{\sgn+\cindex{e}}{S} = \szgoodfrac{\sgn}{S} - 1$, if
$C(e)\cap S\neq \emptyset$, and
 $\szgoodfrac{\sgn+\cindex{e}}{S} = \szgoodfrac{\sgn}{S}$, otherwise.

Note that Fact~\ref{fact:recursive} along with $\yvec{\sgn}{t} \in
\saop^t(\homog{\atspbal(G)})$ implies that
$\yvec{\sgn}{t+1}$, restricted to $\pop_{t+1}$,
is in $\saop^t(\homog{\atspbal(G)})$ because
it can be written as a convex combination of $\yvec{\sgn}{t}$ and an
integral feasible solution $\onevec{E}{t+1}$.
Equation~\eqref{eqn:bal-case1} proves Lemma \ref{lemma:cond1} because
both $\yvec{\sgn+\cindex{e}}{t}$ and
$\yvec{\sgn}{t+1}$ (restricted to $\pop_{t+1}$) are in
$\saop^t(\homog{\atspbal(G)})$.
 \end{proof}

\begin{lemma}
\label{lemma:cond2}
Suppose that $\yvec{\sgn}{t} \in \saop^t(\homog{\atspbal(G)})$,
for each $\sgn\subseteq\fracset$.
Then for all $e \in E$ and for all $\sgn\subseteq\fracset$ we have
$\yvec{\sgn}{t+1} - e*\yvec{\sgn}{t+1}  \in \saop^t(\homog{\atspbal(G)})$.
\end{lemma}

\begin{proof}
Let $C(e)$ denote the dicycle containing edge $e$, and recall
that $\cindex{e}$ denotes the index of $C(e)$.
If $\cindex{e}\in \sgn\cup\notfracset$,
then we have $\goodfrac{\sgn}{S+e} = \goodfrac{\sgn}{S}, \forall S\in\pop_{t+1}$,
hence, we have $\yvec{\sgn}{t+1} = e*\yvec{\sgn}{t+1}$, and
the lemma follows.

Otherwise, we have $\cindex{e}\in \fracset-\sgn$.
Then, for any $S\in\pop_{t+1}$, Equation~\eqref{eqn:shift} gives
\begin{align}
 (\yvec{\sgn}{t+1} - e*\yvec{\sgn}{t+1})_S &=
 \begin{cases}
  0 &\text{if } C(e)\cap S\neq \emptyset \\
  \frac{1}{t+3}  &\text{if } C(e)\cap S=\emptyset
 \end{cases}\\
 &= \frac{1}{t+3}\onevec{E-C(e)}{t+1}_S
\end{align}
The good-decomposition property of $G$ implies that
$\onevec{E-C(e)}{t+1}$ is
a feasible integral solution of \\ $\saop^t(\atspbal(G))$.
\end{proof}

\begin{lemma}
\label{lemma:base}
We have \qquad
$\displaystyle 
\yvec{\emptyset}{0} \in \saop^{0}(\homog{\atspbal(G)}).
$
\end{lemma}
\begin{proof}
Observe that $\yvec{\emptyset}{0}$ has $|E|+1$ elements,
and $\yvec{\emptyset}{0}_{\emptyset} =1$
(by Definition~\ref{def:balanced});
the other $|E|$ elements are indexed by
the singleton sets of $E$.
For notational convenience, let $y \in \reals^{E}$ denote
the restriction of $\yvec{\emptyset}{0}$ to
indices that are singleton sets;
thus, $y_e = \yvec{\emptyset}{0}_{\{e\}}, \forall e\in{E}$.
By Definition~\ref{def:balanced},
$y_e =1/2$ if $e\in E(C_j)$ where $j\in\fracset$,
and $y_e=1$, otherwise.
%
%
We claim that $y$ is a feasible solution to $\atspbalpolytope(G)$.

$y$ is clearly in $[0, 1]^{E}$.
Moreover, $y$ satisfies the balance-constraint at each vertex because it
assigns the same value (either $1/2$ or 1) to every edge in a dicycle
$C_j$, $\forall j\in\cindset$.


To show feasibility of the cut-constraints, consider any cut $\emptyset\neq
U\subset V$. Since $\textbf{1}^{E}$ is a feasible solution, there
exists an edge $e\in E$ crossing from $U$ to $V-U$. If $e\in E(C_j),
j\in\notfracset$, then we have $y_e=1$, which implies
$y(\delta^{out}(U))=y(\delta^{in}(U))\geq 1$
(from the balance-constraints at the vertices).
 Otherwise, we have $e\in E(C_j), j\in\fracset$.
Applying the good-decomposition property of $G$,
we see that there exists an edge
$e^{\prime} (\neq e) \in E - E(C_{j}) $ such that $e^{\prime}\in
\delta^{out}(U)$, i.e., $|\delta^{out}(U)|\geq 2$. Since $y_e\geq
\frac{1}{2}$ for each $e\in E$, the  cut-constraints
$y(\delta^{in}(U))=y(\delta^{out}(U))\geq 1$ are satisfied.
\end{proof}

The next result presents our first lower bound on the integrality ratio for
the level~$t$ relaxation of the \sa\ procedure starting with the balanced~LP.
The relevant instance is a simple digraph on $\Theta(t)$ vertices;
see Figure~\ref{fig-bal}.
In the next subsection, we present better integrality ratios
using the CGK construction,
but the CGK digraph is not as simple and it has $\Theta(t^t)$ vertices.

\begin{theorem}\label{thm:sgbalIR}
Let $t$ be a nonnegative integer, and let $\epsilon\in\reals$
satisfy $0<\epsilon\ll{1}$.
There exists a digraph on $\nu=\nu(t,\epsilon)=\Theta(t/\epsilon)$ vertices
such that the integrality ratio for
the level~$t$ tightening of the balanced~LP (Bal~LP)
(by the \sa\ system)
is $\geq 1+\frac{1-\epsilon}{2t+3}$.
\end{theorem}

\begin{proof}
Let $G$ be the digraph together with the good~decomposition shown
in Figure~\ref{fig-bal}, and
let the cost of each edge in $G$ be $1$.
We call an edge of $G$ a \emph{thin} edge if it is contained
in a dicycle of length~2;
we call the other edges of $G$ the \emph{thick} edges;
see the illustration in Figure~\ref{fig-bal}.
Consider the metric completion $H$ of $G$.
%
%
It can be seen that the optimal value of
an integral solution of ATSP on $H$
(equivalent to the minimum cost Eulerian subdigraph of $G$)
is $\geq 4\ell+2$,
where $\ell$ is the length of the ``middle path.''
(This can be proved by induction on $\ell$, using similar arguments as
in Cheung~\cite[Claim~3~of~Theorem~11]{cheung05}.)

Given $t$ and $\epsilon$, we
fix $\ell= 2(2t+3)/\epsilon$ to get a digraph $G$
(and its edge costs) from the above family.

By Corollary~\ref{coro:balanced-empty}
the fractional solution $\yvec{\emptyset}{t}$
(Definition \ref{def:balanced})
is in $\saop^t(\atspbalpolytope(G))$:
we have
$\yvec{\emptyset}{t}_e=1$ for each thick edge $e$, and
$\yvec{\emptyset}{t}_e=\frac{t+1}{t+2}$ for each thin edge $e$.
By Section \ref{sec:EVto0}, we can extend $\yvec{\emptyset}{t}$ to a
feasible solution of $\saop^t(\atspbalpolytope(H))$.

Hence, the integrality ratio is
%
%
\[\geq \frac{4\ell + 2}{2\ell + 4 + 2\ell \frac{t+1}{t+2}} \ge
	\frac{2(t+2)}{2t+3} - \frac{2}{\ell} \ge
	1+ \frac{1 - \epsilon}{2t+3}.
\]
\end{proof}


\subsection{CGK (Charikar-Goemans-Karloff) construction}

We briefly explain the CGK~\cite{CGK06} construction and show in
Theorem~\ref{thm:CGK} that the resulting digraph has a good
decomposition. This theorem along with a lemma from \cite{CGK06}
shows that the integrality ratio is $\ge1+\frac{1-\epsilon}{t+1}$
for $t$ rounds of the \sa\ procedure
starting with the Balanced~LP,
for any given $0<\epsilon\ll{1}$, see Theorem~\ref{thm:IR-CGK}.

Let $r$ be a fixed positive integer.
Let $G_{0}$ be the digraph with a single vertex.
Let $G_{1}$ consist of a bidirected path
of $r+2$ vertices, starting at the ``source'' $p$ and ending
at the ``sink'' $q$, whose $2(r+1)$ edges have cost $1$ (see
Figure~\ref{fig:G1}).
We call $E(G_1)$ the \textit{external} edge~set of $G_1$
(we use this in the proof of Lemma~\ref{lem:pre_CGK}).

{
\begin{figure}[h!]
\centering
\subfloat[$G_0$]{
\begin{tikzpicture}
        \node at (0 , 0) [graphnode]{};
        \node at (2,0)[]{};
        \node at (-2,0)[]{};
\end{tikzpicture}
}
\subfloat[$G_1$]{
\begin{tikzpicture}
\foreach \x in {0,1,2,3,4}
{
        \node at (2*\x , 0) [graphnode]{};
}
\foreach \x in {1,2,3,4}
{
        \node at (2*\x -1,0 )[]{$C_{\x}$ };
}

\node at (0-0.3 , 0) []{$p$ };
\node at (8+0.3 , 0) []{$q$ };

\foreach \x in {0,2}
{
        \draw [->,very thick] (2*\x +0.1 , 0.1) .. controls (2*\x +1,0.3) ..
(2*\x + 2-0.1 , 0.1);
        \draw [<-,very thick] (2*\x +0.1 , -0.1) .. controls (2*\x +1,-0.3)
.. (2*\x + 2-0.1 , -0.1);
}
\foreach \x in {1,3}
{
        \draw [->] (2*\x +0.1 , 0.1) .. controls (2*\x +1,0.3) .. (2*\x + 2-0.1 , 0.1);
        \draw [<-] (2*\x +0.1 , -0.1) .. controls (2*\x +1,-0.3) .. (2*\x +
2-0.1 , -0.1);
}
\end{tikzpicture}
}
\caption{$G_0$ and $G_{1}$ for $r=3$}
\label{fig:G1}
\end{figure}
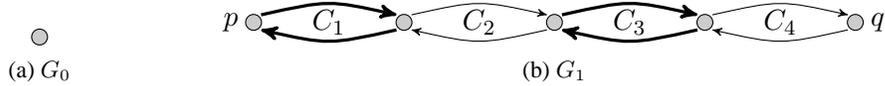

\begin{figure}[h]
\centering
\usetikzlibrary{shapes,snakes}
\tikzstyle{grapheps}=[ellipse, minimum height=.6cm,minimum
width=3.5cm,draw,fill=black!10]
\begin{tikzpicture}[thin,scale=0.8]
\foreach \z in {1,2,3}
{
        \node at (-2 + 5*\z -5, 8.5) [] {$u_\z$};
        \node at (-2 + 5*\z -5, -0.5) [] {$v_\z$};
}
        \draw  [dash pattern=on5pt][->-, very thick] (8 , 8) .. controls (3 ,
10) .. (-2 , 8);
        \draw  [dash pattern=on5pt][->-, very thick] (-2 , 0) .. controls (3
, -2) .. (8 , 0);

\foreach \z in {1,2,3}
{
        \foreach \x in {1,2,...,3}
        {
                \node at (-7 +5*\z, -2*\x + 8) [grapheps]{};
                \node at (-7 +5*\z, -2*\x + 8) []{$G_{k-2}^{(\z,\x)}$};
        }
}

\foreach \z in {0,5,10}
{
        \foreach \x in {1,2,...,3}
        {
                \foreach \y in {0,4}
                {
                        \node at (-\y +\z, 2*\x) [graphnode]{};
                }
        }

        \node at (-2 +\z, 0) [graphnode]{};
        \node at (-2 +\z, 8) [graphnode]{};

        \foreach \x in {1,2}
        {
                \ifthenelse{\z =0 \OR \z= 10}
                {
                        \draw [->-] (0 +\z, 2*\x +0.1 ) to  (0 +\z, 2*\x + 2-0.1 );
                        \draw [->-,very thick]  (-4 +\z, 2*\x + 2-0.1) to (-4 +\z, 2*\x +0.1 );
                }
                {
                        \draw [->-,very thick] (0 +\z, 2*\x +0.1 ) to  (0 +\z, 2*\x + 2-0.1 );
                        \draw [->-]  (-4 +\z, 2*\x + 2-0.1) to (-4 +\z, 2*\x +0.1 );
                }
        }
        \ifthenelse{\z =0 \OR \z= 10}
        {
                \draw [->-,very thick] (-4 +\z, 2) to (-2 +\z, 0) node[pos=0.5,below] {$ $};
                \draw [->-,very thick] (-2 +\z, 8) to (-4 +\z, 6) node[pos=0.5,below] {$ $};
                \draw [->-] (-2 +\z, 0) to (0 +\z, 2*0 + 2-0.1) node[pos=0.5,above] {$ $};
                \draw [->-] (0 +\z, 6) to (-2 +\z, 8) node[pos=0.5,above] {$ $};
        }
        {
                \draw [->-]  (-4 +\z, 2) to  (-2 +\z, 0) node[pos=0.5,below] {$ $};
                \draw [->-]  (-2 +\z, 8) to (-4 +\z, 6) node[pos=0.5,below] {$ $};
                \draw [->-,very thick] (-2 +\z, 0) to (0 +\z, 2*0 + 2-0.1)
node[pos=0.5,above] {$ $};
                \draw [->-,very thick] (0 +\z, 6) to (-2 +\z, 8) node[pos=0.5,above] {$ $};
        }
}
\node at (-6 , 4) [graphnode]{};
\node at (12 , 4) [graphnode]{};
\node at (-6 -0.3, 4) []{$p$ };
\node at (12 +0.3 ,4) []{$q$ };

\draw [->-,very thick] (-6, 4) .. controls (-5,7) .. (-2, 8)
node[pos=0.5,above] {$ $};
\draw [->-,very thick] (-2, 0).. controls (-5,1) .. (-6, 4)
node[pos=0.5,below] {$ $};
\draw [->-] (8, 8) .. controls (11,7) .. (12, 4) node[pos=0.5,above] {$ $};
\draw [->-] (12, 4).. controls (11,1) .. (8, 0) node[pos=0.5,below] {$ $};

\draw [->-] (-2,8) to (3,8) node[pos=0.5,above] {$ $};
\draw [->-] (3,0) to (-2,0) node[pos=0.5,above] {$ $};
\draw [->-,very thick] (3,8) to (8,8) node[pos=0.5,above] {$ $};
\draw [->-,very thick] (8,0) to (3,0) node[pos=0.5,above] {$ $};

\end{tikzpicture}
\caption{$G_{k}$ and $L_k$ for $k\geq 2$ and $r=3$}
\label{fig:G3}
\end{figure}
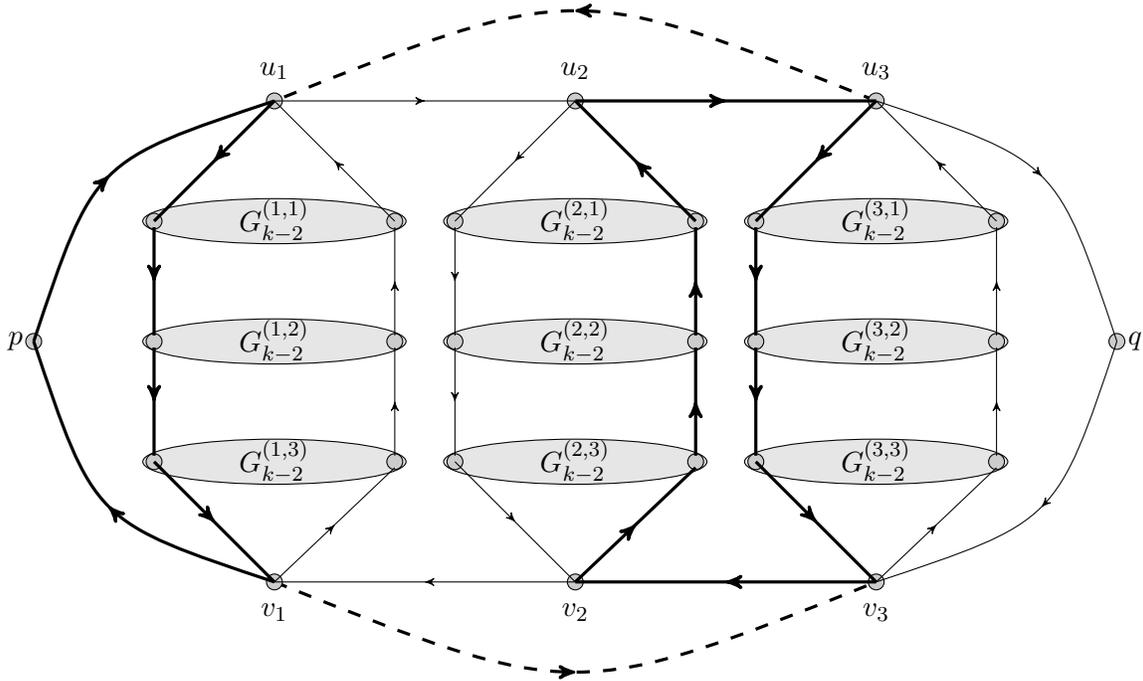
}

For each $k\geq 2$, we construct $G_k$ by taking $r$ copies of
$G_{k-1}$, additional source and sink vertices $p$ and $q$, a dipath
from $p$ to $q$ of $r+1$ edges visiting the sources of the $r$
copies in the order $u_1, u_2, \ldots, u_r$, and another dipath
from $q$ to $p$ of $r+1$ edges visiting the sinks of the $r$ copies
in the order $v_r, v_{r-1}, \ldots, v_1$ where $u_i, v_i$ denote
the source and sink of the $i$-th copy of $G_{k-1}$ (see
Figure~\ref{fig:G3}).
All the new edges have cost $r^{k-1}$. Denote the $i$-th copy of
$G_{k-1}$ by $G_{k-1}^{(i)}$. Let $E_k=E(G_{k})-\cup_{1\leq i\leq r}
E(G_{k-1}^{(i)})$. Let $\{G_{k-2}^{(i, j)}\}_{1\leq j\leq r}$ be the
$r$ copies of $G_{k-2}$ in $G_{k-1}^{(i)}$.
Let $E_{k-1}^{(i)}=E(G_{k-1}^{(i)})-
	\bigcup_{1\leq j\leq r} E(G_{k-2}^{(i,j)})$.
Let $A^{(i)}$ be the dipath from $u_i$ to $v_i$ in $E_{k-1}^{(i)}$
and let $B^{(i)}$ be the dipath from $v_i$ to $u_i$ in $E_{k-1}^{(i)}$.
Let $E^{[r]}_{k-1}=\cup_{1\leq i\leq r} E_{k-1}^{(i)}$.
We call $E_{k}\cup E^{[r]}_{k-1}$ the \textit{external} edge~set of $G_k$.
The other edges form the \textit{internal} edge set of $G_k$.

For each $k\ge2$,
the digraph $L_k$ is constructed from $G_k$ by removing vertices $p$
and $q$, and adding the edges $(u_r,u_1)$ and $(v_1,v_r)$, both of cost
$r^{k-1}$. Let 
$$E^{'}_k=(E_k\cup \{(u_r, u_1), (v_1, v_r)\})-\{ (p,
u_1), (v_1, p), (u_r, q), (q, v_r)\}.$$
 We call $E^{'}_k\cup E^{[r]}_{k-1}$
the \textit{external} edge set of $L_k$. The other edges form the
\textit{internal} edge set of $L_k$.
(Our description of the CGK construction is essentially the same
as in \cite{CGK06}, but they use $s$ and $t$ to denote
the source and sink vertices,
whereas we use $p$ and $q$;
this is to avoid conflict with our symbol $t$
for the number of rounds of the \iSA\ procedure.)

\begin{fact}
\label{fact:Lk-external}
Let $k\ge2$ be a positive integer.
The external edge set of $L_{k}$, i.e., $E_k^{'}\cup E^{[r]}_{k-1}$, can be
partitioned into $r$ dicycles $C_1', \ldots, C_{r}'$ such that
\begin{itemize}
\item[]
	$C_i'=\{(u_{i}, u_{i+1}), (v_{i+1}, v_{i})\}\cup
	B^{(i)}\cup A^{(i+1)}$, for $1\leq i\leq r-1$, and
\item[]
	$C_{r}'=\{(u_{r}, u_1), (v_1, v_{r})\}\cup B^{(r)}\cup A^{(1)}$.
\end{itemize}
Moreover, for each dicycle $C_i'$, $i=1,\dots,r$,
$L_{k}-E(C_i')$ is strongly connected.
\end{fact}

We denote the decomposition of the external edge~set of $L_k$
by $\gdecomplong{E_k^{'}\cup E^{[r]}_{k-1}}{L_k}=\{C_1', \ldots, C_{r}'\}$. 

\begin{fact}
\label{fact:Gk-external}
Let $k\ge2$ be a positive integer.
The external edge set of $G_{k}$, i.e., $E_k \cup E^{[r]}_{k-1}$, can be
partitioned into $r+1$ dicycles $C_0, C_1, \ldots, C_{r}$ such that
\begin{itemize}
\item[]
	$C_i=\{(u_{i}, u_{i+1}), (v_{i+1}, v_{i})\}\cup
	B^{(i)}\cup A^{(i+1)}$, for $1\leq i\leq r-1$,
\item[]
	$C_0=\{(p, u_1), (v_1, p)\}\cup A^{(1)}$, and
\item[]
	$C_{r}=\{(u_{r}, q), (q, v_{r})\}\cup B^{(r)}$.
\end{itemize}
Moreover, for each dicycle $C_i$, $i=0,1,\dots,r$,
$G_{k}-E(C_i)$ has two strongly-connected components,
where one contains the source $p$ and
the other one contains the sink $q$.
\end{fact}

We denote the decomposition of the external edge~set of $G_k$
by $\gdecomplong{E_k\cup{E^{[r]}_{k-1}}}{G_k}=\{C_0,C_1, \ldots, C_{r}\}$. 
Next we identify a structural property 
that will allow us to prove that $L_k$ has a good decomposition. 

\begin{definition}\label{def: pq good decomposition}
We say that $G_k$ has a $p,q$ good decomposition, if the edge~set of $G_{k}$ can be partitioned into
dicycles $C_1, C_2, \ldots, C_N$ such that
for each $1\leq i\leq N$, either
\begin{enumerate}
\item[(1)] $C_i$ consists of external edges, 
and moreover,
$G_{k}-E(C_i)$ has two strongly connected components, one containing
the source $p$ and the other one containing the sink $q$.
\item[(2)] $C_i$ consists of internal edges of $G_k$, and moreover,
$G_{k}-E(C_i)$ is strongly connected.
\end{enumerate}
\end{definition}

\begin{lemma}
\label{lem:pre_CGK}
For all $k\ge1$, $G_k$ has a $p,q$ good decomposition. 
\end{lemma}
\begin{proof}
We prove the result by strong induction on $k$.
For the base cases, consider $G_{1}$ and $G_{2}$.
For $G_{1}$, we take the dicycles $C_1,\dots,C_N$
to be the length~2 dicycles formed by two anti-parallel edges;
thus, $N={r+1}$ (see Figure \ref{fig:G1}).
For $G_{2}$, we use the decomposition of the external edge~set given
by Fact~\ref{fact:Gk-external}.

For the induction step, we have $k\ge3$;
we assume that the statement holds for $1,2,\dots,k-1$ and
prove that it holds for $k$.
By the induction hypothesis, for each $1\leq{i,j}\leq{r}$,
we know that $G^{(i, j)}_{k-2}$ has
a $p,q$~good decomposition
$\gdecomp{E(G^{(i,j)}_{k-2})}=\{C^{(i,j)}_1, C^{(i,j)}_2, \ldots,
	C^{(i,j)}_{N_{(i,j)}}\}$.
Consider the decomposition of $E(G_k)$ into edge-disjoint dicycles
given by
$\gdecompsymbol=\gdecomplong{E_k\cup{E^{[r]}_{k-1}}}{G_k}\cup
	\bigcup_{1\le{i,j}\le{r}} \gdecomp{E(G^{(i,j)}_{k-2})}$.
We claim that $\gdecompsymbol$ is a $p,q$~good decomposition of $G_k$.
Clearly, for $C\in \gdecompsymbol$ such that
$E(C)\subseteq E_k\cup{E^{[r]}_{k-1}}$,
we are done by Fact~\ref{fact:Gk-external}.
Now, consider one of the other dicycles $C\in \gdecompsymbol$;
thus $C$ consists of some internal edges of $G_k$.
Then, there exists an $i$ and $j$ ($1\leq{i,j}\leq{r}$)
such that $C\in \gdecomp{E(G^{(i,j)}_{k-2})}$.
We have two cases, since either condition~(1) or~(2)
of $p,q$~good decomposition of $G^{(i,j)}_{k-2}$ applies to $C$.
In the first case,
$G^{(i, j)}_{k-2}-E(C)$ has two strongly connected components,
where one contains the source $p^{(i, j)}$ of
$G^{(i, j)}_{k-2}$ and
the other one contains the sink $q^{(i, j)}$ of
$G^{(i, j)}_{k-2}$.
Note that the external edge~set of $G_k$ ``strongly connects''
$p^{(i, j)}$ and $q^{(i, j)}$, hence, $G_k-E(C)$ is strongly connected.
In the second case, $G^{(i,j)}_{k-2}-E(C)$ is strongly connected;
then clearly, $G_k-E(C)$ is strongly connected.
Thus $\gdecompsymbol$ is a $p,q$~good decomposition of $G_k$.
\end{proof}

\begin{theorem}
\label{thm:CGK}
For $k\geq 2$,
$L_{k}$ has a good decomposition
with witness~set $\fracset$ such that $\fracset=\cindset$,
i.e. every edge in any cycle in the decomposition can be assigned a fractional value. 
%
\end{theorem}
\begin{proof}
Let $\gdecomplong{E_k^{'}\cup E^{[r]}_{k-1}}{L_k}$
be the decomposition of the external edge~set of $L_{k}$
given by Fact~\ref{fact:Lk-external}.
If $k=2$, then we are done (we have a good decomposition of $L_{k}$
with $\fracset=\cindset$).
Otherwise, we use the decomposition
$\gdecompsymbol=\gdecomplong{E_k^{'}\cup{E^{[r]}_{k-1}}}{L_{k}} \cup
	\bigcup_{1\le{i,j}\le{r}} \gdecomp{E(G^{(i,j)}_{k-2})}$,
where $\gdecomp{E(G^{(i,j)}_{k-2})}$ is a
$p,q$~good decomposition of $G^{(i,j)}_{k-2}$.
Using similar arguments as in the proof of Lemma~\ref{lem:pre_CGK},
it can be seen that $\gdecompsymbol$ is a good decomposition with
$\fracset=\cindset$.
\end{proof}

\begin{lemma}[Lemma~3.2\cite{CGK06}]
\label{lem:CGK}
%
For $k\geq 2$ and $r\geq 3$, the minimum cost of
the Eulerian subdigraph of $L_k$ is $\geq (2k-1)(r-1)r^{k-1}$.
\end{lemma}

\begin{theorem}
\label{thm:IR-CGK}
Let $t$ be a nonnegative integer, and let $\epsilon\in\reals$
satisfy $0<\epsilon\ll{1}$.
There exists a digraph on
$\nu=\nu(t,\epsilon)=O((t/\epsilon)^{(t/\epsilon)})$
vertices such that the integrality ratio for
the level~$t$ tightening of the balanced~LP for ATSP (Bal~LP)
(by the \sa\ system)
is $\geq 1+\frac{1-\epsilon}{t+1}$.
\end{theorem}

\begin{proof}
Given $t$ and $\epsilon$, we apply the CGK construction
with $k=r= 5(t+1)/\epsilon$ to get the digraph $L_k$
and its edge costs. Let $H_k$ be the metric completion of $L_k$.

We know from CGK \cite{CGK06} that
the total cost of the edges in $L_k$ is $\leq 2k(r+1)r^{k-1}$.
By Theorem~\ref{thm:CGK}, $L_k$ has a good decomposition $C_1,\dots,C_N$
such that each of the dicycles $C_j$ has its index in the witness
set $\fracset$ (informally, each edge is assigned to a fractional
dicycle).
Hence, Corollary~\ref{coro:balanced-empty}
implies that the fractional solution that
assigns the value $\frac{t+1}{t+2}$ to (the variable of) each edge
is feasible for $\saop^t(\homog{\atspbalpolytope(L_k)})$. By Section \ref{sec:EVto0}, this feasible solution can
be extended to a feasible solution in $\saop^t(\homog{\atspbalpolytope(H_k)})$.

Then, using Lemma~\ref{lem:CGK},
we see that the integrality ratio of
$\saop^t(\homog{\atspbalpolytope(H_k)})$ is
\begin{align*}
\geq &
\frac{(2k-1)(r-1)r^{k-1}}{(\frac{t+1}{t+2})2k(r+1)r^{k-1}}
%
%
	= 1+\frac{1}{t+1}-\frac{5r-1}{\frac{t+1}{t+2}(r+1)(2r)}
	\geq 1+\frac{1}{t+1}-\frac{5}{\frac{t+1}{t+2}}\frac{1}{(2r)}
\\
\geq & 1+ \frac{1-\epsilon}{t+1}. \nonumber
\end{align*}
\end{proof}


\section{\iSA\ applied to the standard (DFJ~LP) relaxation of ATSP}

Let $G = (V,E)$ be a strongly connected digraph
that has a good decomposition, and moreover,
has both indegree and outdegree $\le2$ for every vertex.
We use the same notation as in Section~\ref{section:balanced},
i.e., $C_1, C_2, \ldots, C_N$ denote the
edge disjoint dicycles of the decomposition,
and there exists $\fracset\subseteq\cindset=\{1,\dots,N\}$
such that $\fracset$ is nonempty and
$G-E(C_j)$ is strongly connected for all $j\in\fracset$.

We define a splitting operation that splits every vertex that
has indegree~2 (and outdegree~2)
into two vertices (along with some edges);
our definition depends on the given good decomposition of the digraph.
The purpose of the splitting operation will be
clear from Fact \ref{fact:CorrectDeg}.

\textbf{Splitting Operation:} Let  $v\in V(G)$ whose indegree and
outdegree is $2$. Suppose $C_i, C_j$ are the dicycles in the
good decomposition going through $v$. Let $e_{i1}=(v_{i1}, v),
e_{j1}=(v_{j1}, v)$ and $e_{i2}=(v, v_{i2}), e_{j2}=(v, v_{j2})$
be the edges in $\delta^{in}(v)$, $\delta^{out}(v)$, respectively,
where $e_{i1}, e_{i2}\in C_i$ and $e_{j1}, e_{j2}\in C_j$.  We split
$v$ into $v^{u}, v^{b}$ as follows:

\begin{itemize}
\item Replace $e_{i1}, e_{j2}$ by $e_{i1}^{new}=(v_{i1},v^{u}),
	e_{j2}^{new}=(v^{u}, v_{j2})$ (the new edges are called solid edges)
\item Replace $e_{i2}, e_{j1}$ by $e_{i2}^{new}=(v^{b},v_{i2}),
	e_{j1}^{new}=(v_{j1}, v^{b})$ (the new edges are called solid edges)
\item Add the auxiliary edges (also called dashed edges)
	$e_{0}=(v^{b}, v^{u}), e'_{0}=(v^{u}, v^{b})$.
\end{itemize}

See Figure~\ref{fig:vertexsplit} for an illustration.

We obtain $G^{new}=(V^{new}, E^{new})$ from $G$ by applying the splitting operation to
every vertex in $G$ whose indegree and outdegree is $2$.
We map each dicycle $C_j$, $j\in\cindset$, of $G$
to a set of edges of $G^{new}$ that we call a cycle
and that we will (temporarily) denote by $C_j^{new}$.
We define $C_j^{new}$ to be the following set of edges:
for every edge of $C_j$, its image (in $G^{new}$) is in $C_j^{new}$;
moreover, for every splitted vertex $v$ of $G$ incident to $C_j$,
note that one of $v^{u}$ or $v^{b}$ (the two images of $v$)
is the head of one of the two edges of $C_j^{new}$ incident to $\{v^{u},v^{b}\}$,
and one of the two auxiliary edges $e_{0},e'_{0}$
has its head at the same vertex;
we place this auxiliary edge also in $C_j^{new}$.
For example, in Figure~\ref{fig:vertexsplit},
the cycle $C_i^{new}$ contains the edges
$e_{i1}^{new}$ (image of $e_{i1}$),
$e_{i2}^{new}$ (image of $e_{i2}$), and the auxiliary edge $e_{0}$,
whereas the cycle $C_j^{new}$ contains the edges
$e_{j1}^{new}$, $e_{j2}^{new}$, and
the auxiliary edge $e'_{0}$.

In what follows,
we simplify the notation for the cycles of $G^{new}$
to $C_j$ (rather than $C_j^{new}$);
there is some danger of ambiguity, but the context will resolve this.
We denote the set of auxiliary edges (also called the dashed edges)
of a cycle $C_j=C_j^{new}$ by $D(C_j)$,
and we denote the set of remaining edges of $C_j=C_j^{new}$ by $E(C_j)$.
Note that 
$E^{new} = E(G^{new}) =
\bigcup_{j\in \cindset} (E(C_j) \cup D(C_j))$.
Clearly, there is a bijection between the edges of $E(C_j)=E(C_j^{new})$
in $G^{new}$ and the edges of $E(C_j)$ in $G$.
Also, observe that in $G^{new}$,
the dashed edges are partitioned among the cycles $C_j^{new},j\in\cindset$.


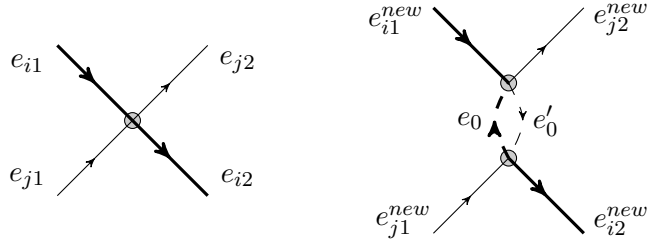
\begin{figure}
\begin{center}
\begin{tikzpicture}
\node at (0,0) [graphnode]{$ $};
\draw [->-, very thick] (-1,1) to node[above]{$e_{i1}\hspace{1.8cm}$} (0,0);
\draw [->-, very thick] (0,0) to node[below]{$\hspace{1.8cm}e_{i2}$} (1,-1);
\draw [->-] (-1,-1) to node[below]{$e_{j1}\hspace{1.8cm}$} (0,0);
\draw [->-] (0,0) to node[above]{$\hspace{1.8cm}e_{j2}$} (1,1);

\node at (5,0.5) [graphnode]{};
\node at (5,-0.5) [graphnode]{};
\draw [->-, very thick] (4,1.5) to node[above]{$e_{i1}^{new}\hspace{2cm}$} (5,0.5);
\draw [->-, very thick] (5,-0.5) to node[below]{\hspace{2cm}$e_{i2}^{new}$} (6,-1.5);
\draw [->-] (4,-1.5) to node[below]{$e_{j1}^{new} \hspace{1.8cm}$} (5,-0.5);
\draw [->-] (5,0.5) to node[above]{$\hspace{2cm}e_{j2}^{new}$} (6,1.5);

\draw [dash pattern=on5pt][->-, very thick] (5,-0.5) .. controls (4.75,0) .. (5, 0.5) node[pos=0.5,left] {$e_{0}$};
\draw [dash pattern=on5pt][->-] (5,0.5) .. controls (5.25,0) .. (5,-0.5) node[pos=0.5,right] {$e'_{0}$};

\end{tikzpicture}
\end{center}
\caption{ An illustration of the vertex splitting operation used for
	mapping $G$ to $G^{new}$.}
\label{fig:vertexsplit}
\end{figure}

\begin{fact}\label{fact:CorrectDeg}
Consider a digraph $G=(V,E)$ that has a good decomposition, and consider
$x\in \reals^{E}$ such that
\textbf{(1)}~$0 \leq x\leq 1$, \quad
\textbf{(2)}~for every dicycle $C_j$, $j\in\cindset$,
	$x_e$ is the same for all edges $e$ of $C_j$, and \quad
\textbf{(3)}~for every vertex $v$ with indegree = 1 = outdegree,
$x(\delta^{in}(v))=x(\delta^{out}(v))=1$.
Then, for the digraph $G^{new}=(V^{new},E^{new})$ obtained by
applying the splitting operations, there exists
$x^{new}\in \reals^{E^{new}}$ such that
$0 \leq x^{new}\leq 1$, and
$x^{new}(\delta^{in}(v))=x^{new}(\delta^{out}(v))=1, \forall v\in V^{new}$.
\end{fact}
\begin{proof}
For each $j\in\cindset$, we consider the dicycle $C_j$.
Let $\alpha_j$ be the $x$-value associated with the dicycle $C_j$ of $G$,
i.e., $x_e=\alpha_j, \forall e\in E(C_j)$.
Then, in $x^{new}$ and $G^{new}$,
we fix $x_e=\alpha_j, \forall e\in E(C_j)=E(C_j^{new})$, and
we fix $x_e=(1-\alpha_j), \forall e\in D(C_j)=D(C_j^{new})$.
It can be seen that $x^{new}$ satisfies the given conditions.
\end{proof}

\begin{definition}
Consider the digraph $G^{new}$.
For any $j\in \fracset$, let $\tour(j) :=
	D(C_j) \cup \bigcup_{i\in (\cindset-j)} E(C_i)$.
\end{definition}

Thus $\tour(j)$
consists of all the solid edges except those in $C_j$ together with
all the dashed edges of $C_j$.
Note that each vertex in $G^{new}$ has exactly one incoming edge and
exactly one outgoing edge in $\tour(j)$. Thus $\tour(j)$ forms a set of
vertex-disjoint dicycles that partition $V^{new}$.

\begin{definition}
\label{def:nicedigraph}
Let $G$ be a digraph with indegree and outdegree $\leq2$
at every vertex, and suppose that $G$ has a good decomposition
with witness set $\fracset$.
Let $G^{new}$ be the digraph obtained by applying splitting operations
to $G$ and its good decomposition.
Then $G$ is said to have the \textit{good tours} property if
$\tour(j)$ is connected
(i.e., $\tour(j)$ forms a Hamiltonian dicycle of $G^{new}$)
for each $j\in\fracset$.
\end{definition}


{
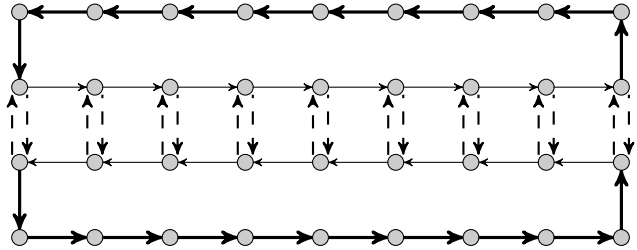
\begin{figure}[htb]
\centering
\begin{tikzpicture}
\foreach \x in {0,1,2,...,8}
{
	\node at (\x , 0) [graphnode]{};
	\node at (\x , 1) [graphnode]{};
	\node at (\x , 2) [graphnode]{};
	\node at (\x , 3) [graphnode]{};
}
\foreach \x in {0,1,2,...,7}
{
	\draw [->, very thick] (\x +0.1 , 0) to node[auto]{$ $} (\x + 1-0.1 , 0);
	\draw [<-] (\x +0.1, 1) to node[auto]{$ $} (\x + 1-0.1 , 1);
	\draw [->] (\x +0.1, 2) to node[auto]{$ $} (\x + 1-0.1 , 2);
	\draw [<-, very thick] (\x +0.1, 3) to node[auto]{$ $} (\x + 1-0.1 , 3);
}

\foreach \y in {0,2}
{
	\draw [<-, very thick] (0 , \y +0.1) to node[auto]{$ $} (0 , \y +1-0.1);
	\draw [->, very thick] (8 , \y +0.1) to node[auto]{$ $} (8 , \y +1-0.1);
}

\foreach \x in {0,1,2,...,7,8}
{
	\draw[dash pattern=on5pt] [->, thick] (\x -0.1 , 1+0.1) to node[auto]{$ $} (\x -0.1, 2-0.1);
	\draw [dash pattern=on5pt][<-, thick] (\x +0.1, 1+0.1) to node[auto]{$ $} (\x +0.1, 2-0.1);
}

\end{tikzpicture}
\caption{Digraph from Figure~\ref{fig-bal} after the splitting operation}
\label{fig-graph}
\end{figure}

\begin{figure}[htb]

\begin{center}
\subfloat[ ]{

\begin{tikzpicture}
\node at (0,0.5) [graphnode]{};
\node at (1,0.5) [graphnode]{};
\node at (3,0)[]{};
\node at (-2,1)[]{};
\draw [->-] (0 +0.1 , 0.5+0.1) .. controls(0 + 0.4, 0.75) .. node[auto]{$ $} (0 + 1-0.1 , 0.5+0.1);
\draw [->-] (0 + 1-0.1 , 0.5-0.1) .. controls(0 + 0.4, 0.25) .. node[auto]{$ $} (0 +0.1 , 0.5-0.1);
\end{tikzpicture}
}
\subfloat[ ]{
\begin{tikzpicture}
\node at (3,0) [graphnode]{};
\node at (3,1)[graphnode]{};
\node at (4,0)[graphnode]{};
\node at (4,1)[graphnode]{};
\draw [<-] (3 +0.1, 0) to node[below]{$ $} (4 -0.1 , 0);
\draw [->] (3 +0.1, 1) to node[auto]{$ $} (4 -0.1 , 1);
\draw[dash pattern=on5pt] [->, thick] (4 -0.1 , 0+0.1) to node[right]{$ $} (4 -0.1, 1-0.1);
\draw [dash pattern=on5pt][<-, thick] (3 +0.1, 0+0.1) to node[auto]{$ $} (3 +0.1, 1-0.1);
\node at (5,0)[]{};
\node at (2,0)[]{};
\end{tikzpicture}
}
\caption{ Transforming a dicycle $C_j$ formed by
an anti-parallel pair of thin edges in Figure~\ref{fig-bal} to
$C_j^{new}$ by the splitting operation.}
\label{fig-square}
\end{center}
\end{figure}
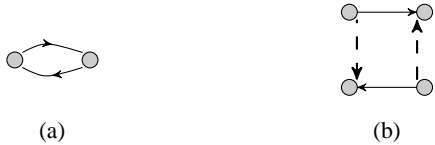

\begin{figure}[htb]
\centering
\begin{tikzpicture}
\foreach \x in {0,1,2,...,8}
{
	\node at (\x , 0) [graphnode]{};
	\node at (\x , 1) [graphnode]{};
	\node at (\x , 2) [graphnode]{};
	\node at (\x , 3) [graphnode]{};
}
\foreach \x in {0,1,2,...,7}
{
	\draw [->, very thick] (\x +0.1 , 0) to node[auto]{$ $} (\x + 1-0.1 , 0);
	\draw [<-, very thick] (\x +0.1, 3) to node[auto]{$ $} (\x + 1-0.1 , 3);
}
\foreach \x in {0,1,2,3,5,6,7}
{
	\draw [<-] (\x +0.1, 1) to node[auto]{$ $} (\x + 1-0.1 , 1);
	\draw [->] (\x +0.1, 2) to node[auto]{$ $} (\x + 1-0.1 , 2);
}
\foreach \y in {0,2}
{
	\draw [<-, very thick] (0 , \y +0.1) to node[auto]{$ $} (0 , \y +1-0.1);
	\draw [->, very thick] (8 , \y +0.1) to node[auto]{$ $} (8 , \y +1-0.1);
}
	\draw [dash pattern=on5pt] [<-, thick] (4 , 1 +0.1) to node[auto]{$e$} (4 , 2-0.1);
	\draw [dash pattern=on5pt] [->, thick] (5 , 1 +0.1) to node[auto]{$ $} (5 , 2-0.1);

\end{tikzpicture}
\caption{$tour(e)$}
\label{fig-tour}
\end{figure}
}

\subsection{Certifying a feasible solution}

In what follows, we assume that $G$ is a digraph
that satisfies the conditions stated in
Definition~\ref{def:nicedigraph}.
We focus on the digraph
$G^{new}$ obtained by applying splitting operations to $G$;
observe that $G^{new}$ depends on $G$ as well as on
the given good decomposition of $G$. Let $\atspdfj(G^{new})$ be the associated cone of $\atspdfjpolytope(G^{new})$.

Let $E$ denote the set of images of the edges of $G$ (the solid edges),
and let $D$ denote the set of auxiliary edges (the dashed edges).
Given $S\subseteq E$ and $\sgn \subseteq \fracset$,
let $\goodfrac{\sgn}{S}$ denote the
set of indices $j\in \fracset - \sgn$ such that
$E(C_j)$ intersects $S$,
and let $\szgoodfrac{\sgn}{S}$ denote the size of this set;
thus, $\szgoodfrac{\sgn}{S}$ denotes the number of ``fractional cycles''
that intersect $S$ in the solid edges.

Note that each (solid or dashed) edge $e$ is in a unique cyle $C(e)$;
let $\cindex{e}$ denote the index of $C(e)$ in $\cindset$;
if $\cindex{e}\in\fracset-\sgn$, then we use $\tour(e)$ to
denote $\tour(\cindex{e})$.



Let $t$ be a nonnegative integer.
We define the feasible solution $\zveconly$ for
the level~$t$ tightening of the DFJ-LP
(of ATSP, by the \iSA\ system)
as follows:
\begin{definition}\label{def:dfj-feasol}
For a nonnegative integer $t$ and
for any subset $\sgn$ of $\fracset$,
let $\yvec{\sgn}{t}$ be a vector indexed by the elements of
$\mathcal{P}_{t+1}$ and defined as follows:
\label{def:dfj-zvec}
\begin{align}
(\zvec{\sgn}{t})_S =
 \begin{cases}
	\frac{t+2-\szgoodfrac{\sgn}{S}}{t+2} & \textup{if~} S\cap D = \emptyset
		\qquad \textup{($S$ has no dashed edges)} \\
	\frac{1}{t+2} & \textup{if~} S\cap D \not= \emptyset \textup{~and~}
		\exists i\in \fracset - \sgn ~:~ \tour(i) \supseteq S \\
		& \hspace{-0.7cm}
		\qquad
 \textup{\textit{($S$ contains some dashed edges and is contained in a tour)}} \\
	0 & \textup{otherwise}
 \end{cases}
\end{align}
\end{definition}

Observe that the second case applies when
the set $S$ has one or more dashed edges,
and moreover, $S$ is contained in a
$\tour(i)$, $i\in\fracset - \sgn$;
also, observe that there is at most one tour that contains $S$,
because the dashed edges are partitioned among the cycles $C_j,j\in\cindset$,
so each dashed edge in $S$ belongs to a unique tour.

\begin{theorem}
\label{thm:dfj-feasible}
Let $G=(V,E)$ be a strongly connected digraph that
has a good decomposition with witness set~$\fracset$,
and moreover,
has (i)~both indegree and outdegree $\le2$ for every vertex,
and (ii)~satisfies the ``good tours'' property.
Then, for any nonnegative integer $t$, and
any $\sgn\subseteq\fracset$ with $|\sgn|\le|\fracset| - (t+2)$,
we have
\[ \zvec{\sgn}{t} \in \saop^t(\homog{\atspdfjpolytope(G^{new})}).
\]
\end{theorem}

\begin{proof}
Note that $\zvec{\sgn}{t}_\emptyset=1$ by Definition~\ref{def:dfj-feasol}.
Thus, we only need to prove $\zvec{\sgn}{t} \in \saop^t(\homog{\atspdfj(G^{new})})$.
The proof is by induction on $t$.
The base case is important, and it follows easily from the
good decomposition property and the ``good tours'' property of $G$.
This is done in Lemma~\ref{lemma:dfj-basecase} below, where we show that
%
%
$\zvec{\sgn}{0} \in \saop^0(\atspdfj(G^{new})),
	\;\forall \sgn\subseteq\fracset, |\sgn|\le |\fracset|-2$.

In the induction step, we assume that
$\zvec{\sgn}{t} \in \saop^t(\atspdfj(G^{new}))$ for some integer $t\ge0$
(the induction hypothesis),
and we apply the recursive definition based
on the shift operator, namely,
$\zvec{\sgn}{t+1} \in \saop^{t+1}(\atspdfj(G^{new}))$ iff for each $e\in E^{new}$
\begin{align}
 e * \zvec{\sgn}{t+1} \in \saop^{t}(\atspdfj(G^{new})), \label{eq:dfj-cond1} \\
 \zvec{\sgn}{t+1} - e*\zvec{\sgn}{t+1} \in \saop^{t}(\atspdfj(G^{new})).  \label{eq:dfj-cond2}
\end{align}
Lemma~\ref{lemma:dfj-cond1} (below) proves (\ref{eq:dfj-cond1}) and
Lemma~\ref{lemma:dfj-cond2} (below) proves (\ref{eq:dfj-cond2}).
\end{proof}

The next lemma proves the base case for the induction;
it follows from the ``good tours'' property of the digraph.

\begin{lemma} \label{lemma:dfj-basecase}
\[
	\zvec{\sgn}{0} \in \saop^{0}(\atspdfj(G^{new})), \quad
	\forall \sgn\subseteq\fracset, |\sgn|\le |\fracset|-2
\]
\end{lemma}

\begin{proof}
Note that $\zvec{\sgn}{0}_{\emptyset}=1$.
Let $z$ be the subvector of $\zvec{\sgn}{0}$
on the singleton sets $\{e_i\}$.
We need to prove that $z$ is a feasible
solution of the DFJ~LP.
It can be seen that $z$ is as follows:
if $\cindex{e}\in\fracset-\sgn$, then $z_e=\frac12$,
otherwise, if $e\in E$ ($e$ is a solid edge), then $z_e=1$,
otherwise, if $e\in D$ ($e$ is a dashed edge), then $z_e=0$.
Clearly, $z$ is in $[0, 1]^{E^{new}}$ and satisfies the degree constraints.
Now, we need to verify that $z$ satisfies the cut constraints
in the digraph $G^{new}$.
Consider any nonempty set of vertices $U\not=V$, and the cut $\delta^{out}(U)$.

Observe that $|\fracset-\sgn| \ge2$, hence, there are at least two indices
$i,j$ such that $i,j\in \fracset-\sgn$.
Hence, both $\tour(i)$ and $\tour(j)$ exist;
moreover, every edge $e$ (either solid or dashed) in either
$\tour(i)$ or $\tour(j)$ has $z_e\ge\frac12$.
Clearly, each of $\tour(i)$ and $\tour(j)$ has
at least one edge in $\delta^{out}(U)$.
Let $e_j$ be an edge of $\tour(j)$ that is in $\delta^{out}(U)$.
If $z_{e_j}=1$, then we are done,
since we have $z(\delta^{out}(U))\ge z_{e_j}=1$.
Thus, we may assume $z_{e_j}=\frac12$.
Now, we have two cases.

First, suppose that $e_j$ is a dashed edge.
Then, note that the edge of $\tour(i)$ in $\delta^{out}(U)$,
call it $e_i$, is distinct from $e_j$
(since the tours are disjoint on the dashed edges),
and again we are done, since $z(\delta^{out}(U))\ge z_{e_i}+z_{e_j}\ge1$.

In the remaining case, $e_j\in\tour(j)$ is a solid edge
and $z_{e_j}=\frac12$.
Then, $\cindex{e_j}\in \fracset-\sgn$, and so 
$\tour(e_j)$
exists and it has at least one edge $e'$ in $\delta^{out}(U)$;
moreover, $e'\not=e_j$ because 
$\tour(e_j)$
contains none of the solid edges of the cycle $C_{\cindex{e_j}}$.
Thus, we are done, since $z(\delta^{out}(U))\ge z_{e_j}+z_{e'}\ge1$.
It follows that $z$ staisfies all of the cut constraints.

\end{proof}

The following fact summarizes some easy observations;
this fact is  used in the next lemma.

\begin{fact}
\label{fact:set-intour}
Let $\sgn$ be a subset of $\fracset$.
Suppose that $S$ is not contained in any
$\tour(j)$, $j\in\fracset - \sgn$.
\textbf{(1)}
Then, for any edge $e$,
$S+e$ is also not contained in any
$\tour(j)$, $j\in\fracset - \sgn$.
\textbf{(2)}
Similary, for any index $\indfrac \in\fracset$,
$S$ is not contained in any
$\tour(j)$, $j\in\fracset - (\sgn + \indfrac)$.
\end{fact}
%

\begin{lemma}
\label{lemma:dfj-cond1}
Suppose that for any nonnegative integer $t$ and
any $\sgn'\subseteq \fracset$ with $|\sgn'|\leq |\fracset| - (t+2)$,
we have $\zvec{\sgn'}{t} \in \saop^t(\atspdfj(G^{new}))$. Then for any $\sgn\subseteq \fracset$ with $|\sgn|\leq |\fracset| - (t+3)$,
\begin{align}
e*\zvec{\sgn}{t+1} \in \saop^t(\atspdfj(G^{new})), \quad \forall e\in E^{new}. \nonumber
\end{align}
\end{lemma}

\begin{proof}
For any edge $e$ and any $S \in \pop_{t+1}$,
the definition of the shift operator gives
\[ (e * \zvec{\sgn}{t+1})_S ~=~ \zvec{\sgn}{t+1}_{S+e}
\]

Let $C(e)$ denote the cycle containing edge $e$,
and let $\cindex{e}$ denote the index of $C(e)$ in $\cindset$.

We will show that
\begin{align}
\label{eqn:dfj-case1}
 (e * \zvec{\sgn}{t+1})_S ~=~
 \begin{cases}
  \zvec{\sgn}{t+1}_S	& \textup{if~} e\in E(C_j)
		\textup{~where~} j\in \sgn\cup\notfracset \\
		&
		\qquad \textup{\textit{($e$ is a solid, integral edge)}} \\
  0		& \textup{if~} e\in D(C_j)
		\textup{~where~} j\in \sgn\cup\notfracset \\
		&
		\qquad \textup{\textit{($e$ is a dashed, integral edge)}} \\
  \frac{t+2}{t+3}\zvec{\sgn+\cindex{e}}{t}_S	& \textup{if~} e\in E(C_j)
		\textup{~where~} j\in \fracset - \sgn \\
		&
		\qquad \textup{\textit{($e$ is a solid, fractional edge)}} \\
  \frac{1}{t+3}\onevec{\tour(e)}{t+1}_S	& \textup{if~} e\in D(C_j)
		\textup{~where~} j\in \fracset - \sgn \\
		&
		\qquad \textup{\textit{($e$ is a dashed, fractional edge)}} \\
 \end{cases}
\end{align}
Lemma \ref{lemma:dfj-recursive} (below) shows that
\[
\zvec{\sgn}{t+1}_S=\frac{t+2}{t+3} \zvec{\sgn+\indfrac}{t}_S +
	\frac{1}{t+3}\onevec{\tour(\indfrac)}{t+1}_S, \qquad
		\forall h\in \fracset - \sgn.
\]
Hence, for every edge $e$ (i.e., in every case), 
	$e*\zvec{\sgn}{t+1}$ is in $\saop^t(\atspdfj(G^{new}))$.
\begin{description}{
\item[Case~1.]
$e\in E(C_j) \textup{~where~} j\in \sgn\cup\notfracset$
{($e$ is a solid, integral edge)}.
We apply Definition~\ref{def:dfj-zvec} (the definition of $\zveconly$), and
consider the three cases in it:
\\
\begin{description}{
\item[Subcase~1.1.]
$S\cap D=\emptyset$.
Then we have $(S+e)\cap D=\emptyset$, and moreover,
we have $\szgoodfrac{\sgn}{S} = \szgoodfrac{\sgn}{S+e}$
(the number of ``fractional cycles'' intersecting $S\cap{E}$ and $(S+e)\cap{E}$
is the same, since $e$ is a non-fractional edge).
Hence, $\zvec{\sgn}{t+1}_{(S+e)} = \zvec{\sgn}{t+1}_{S}$.
\\
\item[Subcase~1.2.]
$S\cap D\not=\emptyset$ and
$\exists i\in\fracset-\sgn:\tour(i)\supseteq{S}$.
Then it is clear that $(S+e)\cap D\not=\emptyset$ and $\tour(i)\supseteq{S+e}$,
because $\tour(i)$ contains every solid edge except those in
the fractional cycle $C_i$.
Hence, $\zvec{\sgn}{t+1}_{S+e} = \frac{1}{t+3} = \zvec{\sgn}{t+1}_{S}$.
\\
\item[Subase~1.3.]
$S\cap D\not=\emptyset$ and
$\forall j\in\fracset-\sgn:\tour(j)\not\supseteq{S}$.
Then it is easily seen that both conditions apply to $S+e$ (rather than $S$).
Hence, $\zvec{\sgn}{t+1}_{S+e} = 0 = \zvec{\sgn}{t+1}_{S}$.
}\end{description}

\item[Case~2.]
We have
$e\in D(C_j) \textup{~where~} j\in \sgn\cup\notfracset$
{($e$ is a dashed, integral edge)}.
We apply Definition~\ref{def:dfj-zvec}, noting that
$(S+e)\cap D\not=\emptyset$ and
there exists no index $i\in \fracset-\sgn$ such that $\tour(i)\supseteq{S+e}$
(no ``valid tour'' contains a dashed, integral edge),
hence, $\zvec{\sgn}{t+1}_{S+e} = 0$.

\item[Case~3.]
We have
$e\in E(C_j) \textup{~where~} j\in \fracset-\sgn$
{($e$ is a solid, fractional edge)}.
We apply Definition~\ref{def:dfj-zvec}.
We have two subcases, either
$S\cap D=\emptyset$, or not.
\\
\begin{description}{
\item[Subcase~3.1.]
If $S\cap D=\emptyset$, then $(S+e)\cap D=\emptyset$.
Thus, the analysis is the same as in the previous section; in particular,
see Equation~\eqref{eqn:bal-case1} in the proof of Lemma~\ref{lemma:cond1}.
Hence, we have
$\zvec{\sgn}{t+1}_{S+e} = \frac{t+2}{t+3}\zvec{\sgn+\cindex{e}}{t}_S$.

\item[Subcase~3.2.]
Otherwise, $S\cap D\not=\emptyset$.
Then we have two further subcases:
either there is an $i\in\fracset-\sgn$ with $\tour(i)\supseteq{S}$
or not.
\\
\begin{description}{
\item[Subcase~3.2.1]
Consider the first subcase;
thus, $S\subseteq\tour(i)$ where $i\in\fracset-\sgn$.
Note that $S$ is not contained
in other tours since $S\cap D\not=\emptyset$.
We have two further subcases, either $e\in E(C_i)$ or not.

\begin{description}{
\item[Subcase~3.2.1.1.]
If $e\in E(C_i)$, then $\tour(i)\not\supseteq{(S+e)}$,
hence, $\zvec{\sgn}{t+1}_{S+e} = 0$
(by the last case in the definition of $\zveconly$);
moreover, note that $\tour(i)$ is the unique tour containing $S$
but it is \emph{not} a ``valid tour'' w.r.t.\ $\sgn+\cindex{e}$,
hence, $\zvec{\sgn+\cindex{e}}{t}_S=0$
(by the last case in Definition~\ref{def:dfj-zvec}).

\item[Subcase~3.2.1.2.]
Otherwise, if $e\not\in E(C_i)$, then
$\tour(i)\supseteq{(S+e)}$, and moreover,
$\tour(i)$ \emph{is} a ``valid tour'' w.r.t.\ $\sgn+\cindex{e}$
(since $i\not\in\sgn$ and $i\not=\cindex{e}$),
%
%
hence, we have
$\zvec{\sgn}{t+1}_{S+e} = \frac{1}{t+3} =
  \frac{t+2}{t+3} \frac{1}{t+2} = \frac{t+2}{t+3} \zvec{\sgn+\cindex{e}}{t}_S$
(by the second case in Definition~\ref{def:dfj-zvec},
for both LHS and RHS).
}\end{description}

\item[Subcase~3.2.2.]
Consider the last subcase;
thus, $S\not\subseteq\tour(i)$ for all $i\in\fracset-\sgn$.
Then by Fact~\ref{fact:set-intour}, the same assertion
holds w.r.t.\ $(S+e)$ (rather than $S$),
as well as w.r.t.\ $(\sgn+\cindex{e})$ (rather than $\sgn$).
Hence, we have
$\zvec{\sgn}{t+1}_{S+e} = 0 = \frac{t+2}{t+3} \zvec{\sgn+\cindex{e}}{t}_S$
(by the last case in Definition~\ref{def:dfj-zvec},
for both LHS and RHS).
}\end{description}
} \end{description}

\item[Case~4.]
We have
$e\in D(C_j) \textup{~where~} j\in \fracset-\sgn$
{($e$ is a dashed, fractional edge)}.
We apply Definition~\ref{def:dfj-zvec}, noting that
$(S+e)\cap D\not=\emptyset$.
We have two subcases, either
$\tour(e)\supseteq{S}$, or not.
If $\tour(e)\supseteq{S}$, then the second case of
Definition~\ref{def:dfj-zvec}
together with the fourth case of Equation~\eqref{eqn:dfj-case1}
(the definition of $e*\zveconly$) gives
$\zvec{\sgn}{t+1}_{S+e} = \frac{1}{t+3} =
	\frac{1}{t+3} \onevec{\tour(e)}{t+1}_S$.
Otherwise, $\tour(e)\not\supseteq{S}$, and then we have
$\zvec{\sgn}{t+1}_{S+e} = 0 =
	\frac{1}{t+3} \onevec{\tour(e)}{t+1}_S$;
note that the last case of Definition~\ref{def:dfj-zvec}
applies because $\tour(e)$ is
the unique ``valid tour'' that could contain $e$.
}\end{description}
\end{proof}

Lemma~\ref{lemma:dfj-recursive} shows that $\zvec{\sgn}{t+1}$,
restricted to $\pop_{t+1}$,
is in $\saop^t(\atspdfj(G^{new}))$;
this is used in Lemma~\ref{lemma:dfj-cond1} to show that $e * \zvec{\sgn}{t+1}$
is in $\saop^t(\atspdfj(G^{new}))$.
%

\begin{lemma}
\label{lemma:dfj-recursive}
For any nonnegative integer $t$,
any $S\in\pop_{t+1}$,
any $\sgn\subseteq \fracset$ with $|\sgn|\leq |\fracset| - (t+3)$,
and any $\indfrac\in\fracset-\sgn$,
we have
\begin{align}
\zvec{\sgn}{t+1}_S ~=~ \frac{t+2}{t+3} \zvec{\sgn+\indfrac}{t}_S +
	\frac{1}{t+3}\onevec{\tour(\indfrac)}{t+1}_S
\end{align}
\end{lemma}
\begin{proof}
%
We have $S\subseteq D\cup{E}, |S|\leq{t+1}$.

We apply Definition~\ref{def:dfj-zvec} (the definition of $\zveconly$)
to $\zvec{\sgn}{t+1}$, and we have three cases.

\begin{description}{
\item[Case~1.]
$S\cap{D}=\emptyset$.
Then $\zvec{\sgn}{t+1}_S = \frac{(t+3)-\szgoodfrac{\sgn}{S}} {t+3}$.
For the RHS, we have two subcases, either $\tour(\indfrac)\supseteq{S}$ or not.
In the first subcase, we have $S\cap E(C_{\indfrac})=\emptyset$
(since $\tour(\indfrac)$ contains none of the solid edges of $C_{\indfrac}$),
hence, $\szgoodfrac{\sgn+\indfrac}{S} = \szgoodfrac{\sgn}{S}$,
consequently, the RHS is
$ \frac{t+2}{t+3} \frac{(t+2)-\szgoodfrac{\sgn}{S}} {t+2} +
	\frac{1}{t+3}$,
which is the same as the LHS.
In the other subcase, $\tour(\indfrac)\not\supseteq{S}$.
Then, we have $S\cap{E(C_{\indfrac})}\not=\emptyset$
(because $S\subseteq{E}$ and $\tour(\indfrac)$ contains all solid edges
except those in $C_{\indfrac}$), hence,
$\szgoodfrac{\sgn+\indfrac}{S} = \szgoodfrac{\sgn}{S} - 1$,
and consequently, the RHS is
$ \frac{t+2}{t+3} \frac{(t+3)-\szgoodfrac{\sgn}{S}} {t+2} + 0 =
	\frac{(t+3)-\szgoodfrac{\sgn}{S}} {t+3}$,
which is the same as the LHS.

\item[Case~2.]
$S\cap{D}\not=\emptyset$ and
there exists $j\in \fracset-\sgn$ such that $\tour(j)\supseteq{S}$.
Then $\zvec{\sgn}{t+1}_S = \frac{1} {t+3}$,
by Definition~\ref{def:dfj-zvec}.
For the RHS, we have two subcases, either $j=\indfrac$ or not.
In the first subcase, we have $\zvec{\sgn+\indfrac}{t}_S = 0$,
because $\tour(\indfrac)$ is the unique tour containing $S$
but it is \emph{not} a ``valid tour'' w.r.t.\ $\sgn+\indfrac$,
hence, the last case in Definition~\ref{def:dfj-zvec} applies.
Thus, the RHS is
$ 0 + \frac{1}{t+3} \onevec{\tour(\indfrac)}{t+1}_S = \frac{1}{t+3}$,
which is the same as the LHS.
In the second subcase, $j\not=\indfrac$.
Then, in the RHS, $\zvec{\sgn+\indfrac}{t}_S = \frac{1} {t+2}$,
because $j\in\fracset-(\sgn+\indfrac)$ and $\tour(j)\supseteq{S}$
so the second case in Definition~\ref{def:dfj-zvec} applies.
Moreover, $\onevec{\tour(\indfrac)}{t+1}_S = 0$,
because $j\not=\indfrac$, and $\tour(j)$ is the unique tour containing $S$,
so $\tour(\indfrac)\not\supseteq{S}$.
Thus, the RHS is
$ \frac{t+2}{t+3} \frac{1} {t+2} + 0 =
	\frac{1}{t+3}$,
which is the same as the LHS.

\item[Case~3.]
$S\cap{D}\not=\emptyset$ and
$\tour(j)\not\supseteq{S}$, $\forall j\in \fracset-\sgn$.
Then $\zvec{\sgn}{t+1}_S = 0$.
In the RHS, $\zvec{\sgn+\indfrac}{t}_S = 0$,
by the third case in Definition~\ref{def:dfj-zvec},
since the relevant conditions hold (by Fact~\ref{fact:set-intour}).
Moreover, $\onevec{\tour(\indfrac)}{t+1}_S = 0$,
because $\indfrac\in\fracset-\sgn$ and $\tour(\indfrac)\not\supseteq{S}$.
Thus, the RHS is
$0$,
which is the same as the LHS.
}\end{description}
This completes the proof of the lemma.
\end{proof}

\begin{lemma}
\label{lemma:dfj-cond2}
Suppose that for any nonnegative integer $t$ and
any $\sgn'\subseteq \fracset$ with $|\sgn'|\leq |\fracset| - (t+2)$,
we have $\zvec{\sgn'}{t} \in \saop^t(\atspdfj(G^{new}))$. Then for any $\sgn\subseteq \fracset$ with $|\sgn|\leq |\fracset| - (t+3)$,
\begin{align}
\zvec{\sgn}{t+1}-e*\zvec{\sgn}{t+1} \in \saop^t(\atspdfj(G^{new})), \quad \forall e\in E^{new} \nonumber
\end{align}
\end{lemma}
\begin{proof}By Lemma~\ref{lemma:dfj-cond1} and
Lemma~\ref{lemma:dfj-recursive}, we have for each $e\in E^{new}=E\cup{D}$ and
any $S\in\pop_{t+1}$,
\begin{align}
  ( \zvec{\sgn}{t+1} - e * \zvec{\sgn}{t+1} )_S ~=~
 \begin{cases}
  0		& \textup{if~} e\in E(C_j)
		\textup{~where~} j\in \sgn\cup\notfracset \\
		&
		\qquad \textup{\textit{($e$ is a solid, integral edge)}} \\
  \zvec{\sgn}{t+1}_S
		& \textup{if~} e\in D(C_j)
		\textup{~where~} j\in \sgn\cup\notfracset \\
		&
		\qquad \textup{\textit{($e$ is a dashed, integral edge)}} \\
  \frac{1}{t+3}\onevec{\tour(e)}{t+1}_S
		& \textup{if~} e\in E(C_j)
		\textup{~where~} j\in \fracset - \sgn \\
		&
		\qquad \textup{\textit{($e$ is a solid, fractional edge)}} \\
  \frac{t+2}{t+3}\zvec{\sgn+\cindex{e}}{t}_S
		& \textup{if~} e\in D(C_j)
		\textup{~where~} j\in \fracset - \sgn \\
		&
		\qquad \textup{\textit{($e$ is a dashed, fractional edge)}} \\
 \end{cases}
\end{align}
Hence, in every case, $\zvec{\sgn}{t+1}-e*\zvec{\sgn}{t+1} \in \saop^t(\atspdfj(G^{new}))$
\end{proof}

\begin{theorem}\label{thm:dfj-sgIR}
Let $t$ be a nonnegative integer, and let $\epsilon\in\reals$
satisfy $0<\epsilon\ll{1}$.
There exists a digraph on $\nu=\nu(t,\epsilon)=\Theta(t/\epsilon)$ vertices
such that the integrality ratio for
the level~$t$ tightening of the standard~LP (DFJ~LP)
(for ATSP, by the \sa\ procedure)
is $\geq 1+\frac{1-\epsilon}{2t+3}$.
\end{theorem}

\begin{proof}
Given $t$ and $\epsilon$, we
fix $\ell= 2 (2t+3)/\epsilon$ to get a digraph $G$ shown in
Figure~\ref{fig-bal} where $\ell$ is the length of the ``middle path''.
Let the cost of each  edge in $G$ be~$1$.
Then we construct $G^{new}$ from $G$.
We keep the cost of edges in $G$ to be~$1$ and fix
the cost of new edges to be~$0$.
See Figure~\ref{fig-graph}; each solid edge has cost~$1$ and
each dashed edge has cost~$0$.
In the proof of Theorem \ref{thm:sgbalIR}, we claimed that
the minimum cost of an Eulerian subdigraph of $G$ is $\geq 4\ell+2$.
It can be seen that the minimum cost of an Eulerian subdigraph of
$G^{new}$ is $\geq 4\ell+2$.
(To see this, take an Eulerian subdigraph of $G^{new}$,
then contract all dashed edges contained in it,
to get an Eulerian subdigraph of $G$ of the same cost.)
Let $H$ be the metric completion of $G^{new}$.
Then, the optimal value of the integral solution in
$\saop^t(\atspdfjpolytope(H))$ is $\geq 4\ell+2$.
%
%

Now we invoke Theorem~\ref{thm:dfj-feasible}, according to which
the fractional solution $\yvec{\emptyset}{t} $
(Definition \ref{def:dfj-feasol}) is in $\saop^t(\atspdfjpolytope(G^{new}))$;
see Figure~\ref{fig-graph};
we have
$\yvec{\emptyset}{t}_e=1$ for each solid, thick edge $e$
(the solid edges of the outer cycle),
$\yvec{\emptyset}{t}_e=\frac{t+1}{t+2}$ for each solid, thin edge $e$
(the solid edges of the middle paths), 
while the value of the dashed edges do not contribute to the value of the objective.
By Section \ref{sec:EVto0}, this feasible solution can
be extended to a feasible solution in $\saop^t(\atspdfjpolytope(H))$.


Hence, the integrality ratio of $\saop^t(\atspdfjpolytope(H))$ is
%
%
\[\geq \frac{4\ell + 2}{2\ell + 4 + 2\ell \frac{t+1}{t+2}} \ge
	\frac{2(t+2)}{2t+3} - \frac{2}{\ell} \ge
	1+ \frac{1 - \epsilon}{2t+3}.
\]
\end{proof}


\section{Path ATSP}

Let $G=(V,E)$ be a digraph with nonnegative edge costs $c$, and let $p$
and $q$ be two distinguished vertices. We define
$\PATSPpolytope_{p,q}(G)$ to be the polytope of the following LP that has a variable $x_e$ for each edge $e$ of $G$:

\begin{align}
\textup{minimize} & \sum_{e} c_e x_e & \notag \\
\textup{subject to} && \notag \\
  & \xin{S} \geq 1,& \quad \forall S:~\emptyset \subset S \subseteq V-\{p\} \notag  \\
  & \xout{S} \geq 1,& \quad \forall S:~\emptyset \subset S \subseteq V-\{q\} \notag  \\
  & \xin{\{v\}}  = 1,& \quad \forall v \in V-\{p\} \notag  \\
  & \xout{\{v\}} = 1,& \quad \forall v \in V-\{q\} \notag  \\
  & \xin{\{p\}}  = 0,& \notag  \\
  & \xout{\{q\}}  = 0,& \notag  \\
  & \textbf{0} \leq \bx \leq \textbf{1} & \notag
\end{align}

In particular, when $G$ is a complete digraph with metric costs, the
above LP is the standard relaxation
for the $p$-$q$ path~ATSP, which is to compute a Hamiltonian (or,
spanning) dipath from $p$ to $q$ with minimum cost in the complete
digraph with metric costs. For  $\PATSPpolytope_{p,q}(G)$, we denote the
associated cone by $\PATSP_{p, q}(G)$.

(In the literature, the notation for the two distinguished vertices is
$s,t$, but we use $p,q$ to avoid conflict with our symbol $t$ for the
number of rounds of the \iSA\ procedure.)

An $(p, q)$\textit{-Eulerian subdigraph} $\overline G$ of $G$ is $V$ together
with a collection of edges of $G$ with multiplicities such that (i) for
any $v\in V-\{p, q\}$, the indegree of $v$ equals its outdegree and
(ii) the outdegree of $p$ is larger than its indegree by $1$ and the
indegree of $q$ is larger than its outdegree by $1$ and (iii) $\overline G$
is weakly connected (i.e., the underlying undirected graph is
connected). The $p$-$q$ path~ATSP on the metric completion $H$ of $G$
is equivalent to finding a minimum cost $(p, q)$-Eulerian subdigraph of
$G$.

For any subset $V^{\prime}$ of $V$, we use $G(V^{\prime})$ to denote
the subdigraph of $G$ induced by $V^{\prime}$. As before, we use $\pop_t$ to denote $\pop_t(E)$ (for the groundset $E$). Also, by the \emph{restriction} of $\zveconly$ on $E^{\prime}\subseteq E$ we mean the vector
$\zveconly|_{E^{\prime}}\in \reals^{\pop_{t+1}(E^{\prime})}$ that
is given by
$(\zveconly|_{E^{\prime}})_S=\zveconly_S$
for all $S \in \pop_{t+1}(E^{\prime})$.

\begin{lemma}
\label{reduction}
Let $t$ be a nonnegative integer.
Let $\zveconly\in \saop^t(\homog{\atspdfjpolytope(G)})$.
Suppose that there exists a dipath $Q\subseteq E$
from some vertex $q$ to another vertex $p$
such that
$\zveconly_e=1$ for each $e\in Q$.
Let $V_Q$ denote the set of internal vertices of the dipath $Q$, and
let $G^{\prime} = G(V - V_Q) = G - V_Q$.
Then,
\[
\zveconly|_{E(G^{\prime})}\in \saop^t(\homog {\PATSPpolytope_{p,q}(G^{\prime})}).
\]
\end{lemma}

\begin{proof}
Let $V^{\prime}=V - V_Q$ and let
$E^{\prime}=E(G^{\prime})$, i.e., $G^{\prime}=(V^{\prime}, E^{\prime})$.
The proof is by induction on $t$.
Denote $\zveconly|_{E^{\prime}}$ by $\zveconly^{\prime}$ for short.
Clearly, $\zveconly^{\prime}_{\emptyset}=1$.
Thus, we only need to prove
$\zveconly^{\prime}\in \saop^t(\homog {\PATSP_{p,q}(G^{\prime})})$.

\noindent \textbf{Base case:}
$t=0$.
Let $z$ be the subvector of $\zveconly$ on the singleton sets
$\{e_i\}$, and let $z^{\prime}$ be the subvector of
$\zveconly^{\prime}$ on the singleton sets.

We have to prove that $z^{\prime}$ is a feasible solution of
$\PATSPpolytope_{p, q}(G^{\prime})$. It is easy to see that
$z^{\prime}$ is in $[0, 1]^{E^{\prime}}$ and it satisfies the
degree~constraints. Thus, we are left with the verification of the cut
constraints.  Observe that each positive edge (on which $z$ is
positive) of $G$ with its head (tail) in $V_Q$ has its tail (head) in
$V_Q+q$ ($V_Q+p$).
Let $\emptyset\neq U\subseteq V^{\prime}$.
If $U\subseteq V^{\prime} - \{q\}$, then
observe that every edge in $\delta^{out}_{G}(U)$ has
its head in $V - V_Q - U = V'-U$, hence, we have
$z^{\prime}(\delta^{out}_{G^{\prime}}(U)) = z(\delta^{out}_{G}(U))\geq{1}$.
Similarly, if $U\subseteq V^{\prime} - \{p\}$, then we have
$z^{\prime}(\delta^{in}_{G^{\prime}}(U)) = z(\delta^{in}_{G}(U))\geq1$;
the equation holds because every edge in $\delta^{in}_{G}(U)$ has
its tail in $V - V_Q - U = V'-U$.

\noindent \textbf{Induction Step:}
For $t\geq 0$, we know
$\zveconly^{\prime}\in\saop^{t+1}(\homog{\PATSP_{p,q}(G^{\prime})})$
if and only if for any $e\in E^{\prime}$,
\begin{align}
\label{indStep2}
e*\zveconly^{\prime}	&\in \saop^{t}(\homog{\PATSP_{p,q}(G^{\prime})}) \\\nonumber
\zveconly^{\prime} - e*\zveconly^{\prime}	&\in
	\saop^{t}(\homog{\PATSP_{p,q}(G^{\prime})})\nonumber
\end{align}

Since $\zveconly$ is a feasible solution in $\saop^{t+1}(\homog{\atspdfj(G)})$, we
have
\begin{align}
e*\zveconly	&\in \saop^{t}(\homog{\atspdfj(G)}) \\\nonumber
\zveconly- e*\zveconly	&\in \saop^{t}(\homog{\atspdfj(G)})\nonumber
\end{align}

Note that $e\in E^{\prime}$. For any $S\subseteq E^{\prime}$ such that
$|S|\leq t+1$, we have $(e*\zveconly^{\prime})_S=\zveconly^{\prime}_{S\cup
\{e\}}=\zveconly_{S\cup \{e\}}=(e*\zveconly)_S$. Thus,
$e*\zveconly^{\prime}=(e*\zveconly)|_{E^{\prime}}$.
Similarly, we have  $\zveconly^{\prime} -
e*\zveconly^{\prime}=(\zveconly - e*\zveconly)|_{E^{\prime}}$.
For any $e^{i}\in Q$, since $\zveconly_{e_i}=1$, we have
$\zveconly_{\{e,e_i\}}=\zveconly_e$ (by the definition of the \iSA\ procedure),
hence, we have
\[
(e*\zveconly)_{\{e^{i}\}}=\zveconly_{\{e,e_i\}}=\zveconly_e=(e*\zveconly)_\emptyset.
\]
Similarly,
\[
(\zveconly-e*\zveconly)_{\{e^{i}\}}=\zveconly_{e^i}-\zveconly_{\{e,e_i\}}=1-\zveconly_e=(\zveconly-e*\zveconly)_\emptyset.
\]
\noindent \textbf{Case 1:} $(e*\zveconly)_{\emptyset}=0$.
In this case, all items in $e*\zveconly$ are zero. This implies $e*\zveconly^{\prime}
\in \saop^{t}(\homog{\PATSP_{p, q}(G^{\prime})})$.

\noindent \textbf{Case 2:} $(e*\zveconly)_\emptyset>0$.
In this case, we consider $\frac{e*\zveconly}{(e*\zveconly)_\emptyset}$. Note that
$(\frac{e*\zveconly}{(e*\zveconly)_\emptyset})_{\{e_i\}}=1$ for any $e_i\in Q$ and
$\frac{e*\zveconly}{(e*\zveconly)_\emptyset}\in \saop^{t}(\homog{\atspdfj(G)})$ with value $1$
at the item indexed by $\emptyset$.
By the inductive hypothesis, we have
$\frac{e*\zveconly}{(e*\zveconly)_\emptyset}|_{E^{\prime}}\in
\saop^{t}(\homog{\PATSP_{p, q}(G^{\prime})})$, i.e.,
$\frac{e*\zveconly^{\prime}}{(e*\zveconly^{\prime})_\emptyset}\in
\saop^{t}(\homog{\PATSP_{p, q}(G^{\prime})})$. Thus, $e*\zveconly^{\prime} \in
\saop^{t}(\homog{\PATSP_{p, q}(G^{\prime})})$.
\\
Similarly, we have $\zveconly^{\prime} - e*\zveconly^{\prime}\in
\saop^{t}(\homog{\PATSP_{p, q}(G^{\prime})})$. This completes the proof.
\end{proof}

From the last section, we know that $\zvec{\emptyset}{t}$
(Definition~\ref{def:dfj-feasol}) is in
$\saop^{t}(\homog{\atspdfj(G)})$, where $G$ is defined in
Figure~\ref{fig-graph};
note that $G$ is obtained from the
digraph and the good decomposition given in Figure~\ref{fig-bal}.
The solid edges in $G$ have cost~$1$ and the
dashed edges in $G$ have cost~$0$.

Let $q$ be the right-most vertex in the second row
	(incident to two dashed edges),
let $p$ be the  left-most vertex in the second row
	(incident to two dashed edges),
and let $Q$ be the dipath of solid edges from $q$ to $p$.
%
%
By the definition of $\zvec{\emptyset}{t}$, we have
$\zvec{\emptyset}{t}_{e_i}=1$ for each $e_i\in Q$. Let
$G^{\prime}=G(V^{\prime})$ where $V^{\prime}=V - V_Q$ where
$V_Q$ is the set of internal vertices of the dipath $Q$. The next
result is a direct corollary of Lemma~\ref{reduction}.

\begin{corollary}\label{thm:PATSP-feasol}
We have 
\[
\zvec{\emptyset}{t}|_{E(G^{\prime})}\in
	\saop^t(\homog{\PATSPpolytope_{p, q}(G^{\prime})}), \quad \forall t \in \integers_+.
\]
\end{corollary}

The proof of the next lemma follows from arguments similar to
those in the proof of Theorem \ref{thm:dfj-sgIR}.
\begin{lemma}\label{PATSP:is-lowerbound}
The minimum cost of a $(p, q)$-Eulerian subdigraph of $G^{\prime}$ is
$\geq 3\ell$, where $\ell$ is the number of edges in
the middle path in $G$.
\end{lemma}

\begin{theorem}
Let $t$ be a nonnegative integer, and let $\epsilon\in\reals$
satisfy $0<\epsilon\ll{1}$.
There exists a digraph on $\nu=\nu(t,\epsilon)=\Theta(t/\epsilon)$ vertices
such that the integrality ratio for
the level~$t$ tightening of $\PATSP$ by the \sa\ procedure is
$\ge 1 + \frac{2-\epsilon}{3t+4}$.
\end{theorem}
\begin{proof}
Given $t$ and $\epsilon$, we
fix $\ell= 2 (3t+4)/\epsilon$.
Consider the metric completion $H$ of $G^{\prime}$. By Section
\ref{sec:EVto0}, we can extend the feasible solution from 
Corollary~\ref{thm:PATSP-feasol} to a feasible solution to
$\saop^t(\homog{\PATSPpolytope_{p, q}(H)})$. This gives an upper bound
on the optimal value of a fractional feasible solution to
$\saop^t(\homog{\PATSPpolytope_{p,q}(H)})$. On the other hand,
Lemma \ref{PATSP:is-lowerbound} gives a lower bound on
the optimal value of an integral solution.
Thus, the integrality ratio is at least

\[\frac{3\ell}{\frac{t+1}{t+2}2\ell + l + 2} \ge
	1+\frac{2}{3t+4} - \frac{2}{\ell} \ge
	1+ \frac{2 - \epsilon}{3t+4}.
\]
\end{proof}

\paragraph {Acknowledgements:}
We thank a number of colleagues for useful discussions.
We are grateful to Sylvia Boyd and Paul Elliott-Magwood
for help with ATSP integrality gaps, and to
Levent Tun\c{c}el for sharing his knowledge of the area.


\bibliographystyle{abbrv}
\bibliography{atspbib}

%
%


%
%

\end{document}